\documentclass[11pt]{article}

\usepackage{amsmath,amssymb,amsthm}
\usepackage[margin=1in]{geometry}
\usepackage{graphicx}
\usepackage{algorithm}
\usepackage{algpseudocode}
\usepackage{comment}
\usepackage{color}
\usepackage{bbm}
\usepackage{todonotes}
\usepackage{booktabs}
\usepackage{breakcites}
\usepackage{bm}
\usepackage[numbers,sort&compress]{natbib}
\usepackage{thm-restate}
\usepackage{calrsfs}
\usepackage{tcolorbox}
\usepackage{subcaption}
\usepackage{mathtools}
\usepackage[colorlinks=true, citecolor=orange, linkcolor=orange]{hyperref}
\usepackage[numbers,sort&compress]{natbib}

\theoremstyle{plain}
\newtheorem{theorem}{Theorem}[section]
\newtheorem{lemma}[theorem]{Lemma}

\newtheorem{corollary}[theorem]{Corollary}

\newtheorem{proposition}[theorem]{Proposition}

\theoremstyle{definition}
\newtheorem{definition}[theorem]{Definition}
\newtheorem{example}[theorem]{Example}

\theoremstyle{remark}
\newtheorem{remark}[theorem]{Remark}

\DeclareMathAlphabet{\pazocal}{OMS}{zplm}{m}{n}

\newcommand{\R}{\mathbb{R}}

\newcommand{\E}{\mathop{\mathbb{E}}}

\newcommand{\cF}{\pazocal{F}}
\newcommand{\cM}{\pazocal{M}}

\newcommand{\cI}{\pazocal{I}}

\newcommand{\cL}{\pazocal{L}}

\newcommand{\cP}{\pazocal{P}}

\renewcommand \vec [1]{\bm{#1}}
\newcommand{\norm}[1]{\lVert#1\rVert}

\newcommand{\OPT}{\textsf{OPT}}

\newcommand{\ttheta}{\widetilde{\theta}}

\newcommand{\btheta}{\boldsymbol{\theta}}

\newcommand{\bTheta}{\boldsymbol{\Theta}}

\newcommand{\VCG}{\mathtt{VCG}}
\newcommand{\err}{\mathtt{err}}
\newcommand{\WT}{\mathtt{WT}}

\DeclareMathOperator*{\argmax}{argmax}
\DeclareMathOperator*{\argmin}{argmin}

\usepackage{centernot}

\title{\bf Bicriteria Multidimensional Mechanism Design \\ with Side Information}
\author{Maria-Florina Balcan\thanks{School of Computer Science, Carnegie Mellon University. \texttt{ninamf@cs.cmu.edu}}\and Siddharth Prasad\thanks{Toyota Technological Institute at Chicago. \texttt{sprasad@ttic.edu}}\and Tuomas Sandholm\thanks{Computer Science Department, Carnegie Mellon University, Strategy Robot, Inc., Optimized Markets, Inc., Strategic Machine, Inc. \texttt{sandholm@cs.cmu.edu}}}

\date{}

\begin{document}

\maketitle

\begin{abstract}
We develop a versatile methodology for multidimensional mechanism design that incorporates side information about agents to generate high welfare and high revenue simultaneously. Side information sources include advice from domain experts, predictions from machine learning models, and even the mechanism designer's gut instinct. We design a tunable mechanism that integrates side information with an improved VCG-like mechanism based on {\em weakest types}, which are agent types that generate the least welfare. We show that our mechanism, when its side information is of high quality, generates welfare and revenue competitive with the prior-free total social surplus, and its performance decays gracefully as the side information quality decreases. We consider a number of side information formats including distribution-free predictions, predictions that express uncertainty, agent types constrained to low-dimensional subspaces of the ambient type space, and the traditional setting with known priors over agent types. In each setting we design mechanisms based on weakest types and prove performance guarantees.
\end{abstract}

\section{Introduction}\label{sec:Intro}

Mechanism design is a high-impact branch of economics, computer science, and operations research that studies the implementation of socially desirable outcomes among strategic self-interested agents. Major real-world use cases include combinatorial auctions~\citep{Cramton06:Combinatorial} (\emph{e.g.}, for strategic sourcing~\citep{Sandholm13:Very,Hohner03:Combinatorial} and radio spectrum auctions~\citep{cramton2013spectrum,bichler2017handbook,leyton2017economics}), matching markets~\citep{roth2018marketplaces} (\emph{e.g.}, for housing allocation and ridesharing), project fundraisers, and many more. The two most commonly studied objectives in mechanism design are \emph{welfare} and \emph{revenue}. Welfare maximization, or \emph{efficiency}, is canonically achieved by the classic Vickrey-Clarke-Groves (VCG) mechanism~\citep{Vickrey61:Counterspeculation,Clarke71:Multipart,Groves73:Incentives}. Revenue maximization is a significantly more elusive problem that is only understood in very special cases. The seminal work of~\citet{Myerson81:Optimal} characterized the revenue-optimal mechanism for the sale of a single item in the Bayesian setting, but it is not even known how to optimally sell two items. It is known that welfare and revenue are generally at odds and optimizing one can come at a great expense of the other~\citep{Ausubel06:Lovely,abhishek2010efficiency,anshelevich2016pricing, kleinberg2013ratio, diakonikolas2012efficiency}.  

In this paper we study how \emph{side information} (or \emph{predictions}) about the agents can help with \emph{bicriteria} optimization of both welfare and revenue. Side information can come from a variety of sources that are abundantly available in practice such as predictions from a machine-learning model trained on historical agent data, advice from domain experts, or even the mechanism designer's own gut instinct. Machine learning approaches that exploit the proliferation of agent data have in particular witnessed a great deal of success both in theory~\citep{Conitzer02:Complexity,Likhodedov04:Methods,morgenstern2016learning,munoz2017revenue,balcan2018general, balcan2005mechanism} and in practice~\citep{Edelman07:Internet,Sandholm07:Expressive,Walsh08:Computing,dutting2019optimal,Sandholm13:Very}. In contrast to the typical Bayesian approach to mechanism design that views side information through the lens of a prior distribution over agents, we adopt a prior-free perspective that makes no assumptions on the correctness, accuracy, or source of the side information (though we show how to apply the techniques developed in this paper to that setting as well). A nascent line of work (part of a larger agenda on {\em learning-augmented algorithms}~\citep{mitzenmacher2022algorithms}) has begun to examine the challenge of improving the performance of classical mechanisms with strategic participants when the designer has access to predictions about the agents, but only for fairly specific problem settings~\citep{xu2022mechanism,banerjee2022online,balkanski2022strategyproof,gkatzelis2022improved,agrawal2022learning}. Algorithm and mechanism design with predictions takes a {\em beyond worst case} perspective on performance analysis, the primary motivation being the access to machine-learning predictions about the problem instance that can greatly improve performance ({\em e.g.,} run-time, memory, revenue, welfare, fairness, {\em etc.}) beyond what is possible in the worst case. This is in contrast to the {\em worst-case} perspective that has been traditional in classical algorithm design and theoretical computer science. We contribute to this budding area of {\em mechanism design with predictions} (also called {\em learning-augmented mechanism design}) with a new general side-information-dependent mechanism for a wide swath of multidimensional mechanism design problems that aim for high social welfare and high revenue.

Here we provide a few examples of the forms of side information we consider in various multidimensional mechanism design scenarios. A formal description of the model is in Section~\ref{section:problem_formulation}.

\emph{(1) An auction designer sets a bidder-specific item reserve price of \$10 based on historical knowledge that the bidder is a high-spending bidder and the observation that all other bids received for that item are tightly clustered around \$10. (2) A real-estate agent believes that a particular buyer values a high-rise condominium with a city view three times more than one on the first floor. Alternately, the seller might know for a fact that the buyer values the first property three times more than the second based on set factors such as value per square foot. (3) A homeowner association is raising funds for the construction of a new swimming pool within a townhome complex. Based on the fact that a particular resident has a family with children, the association estimates that this resident is likely willing to contribute at least \$300 if the pool is opened within a block of the resident's house but only \$100 if outside a two-block radius.} 

These are all examples of side information available to the mechanism designer that may or may not be useful or accurate; we include many more such examples throughout to reinforce central concepts. Our methodology allows us to derive welfare and revenue guarantees under different assumptions on the veracity of the side information. Our model of side information involves a general and flexible language wherein nearly any claim of the form {\em ``the joint type profile of the agents satisfies property $P$"} can be expressed and meaningfully used. We study some other forms of side information as well which we describe further in Section~\ref{sec:contributions}.

\subsection{Our contributions}\label{sec:contributions}

Our main contribution is a versatile tunable mechanism that integrates side information about agent types with the bicriteria goal of simultaneously optimizing welfare and revenue. Traditionally it is known that welfare and revenue are at odds and maximizing one objective comes at the expense of the other. Our results show that side information can help mitigate this difficulty. 

\paragraph{Prediction model, type spaces, and weakest-type VCG} In Section~\ref{section:problem_formulation}, we introduce our model of predictors and formally define the components of multidimensional mechanism design. The abstraction of multidimensional mechanism design is a rich language that allows our theory to apply to many real-world settings including combinatorial auctions, matching markets, project fundraisers, and more---we expand on this list of examples further in Section~\ref{section:problem_formulation}. Our model of predictions is highly general and flexible: a prediction is effectively any claim of the form ``agent types satisfy property $P$", modeled via a set-valued function. One key aspect here is that the information conveyed about an agent {\em can depend on the revealed types of all other agents}. For example, a landowner who wants to sell mineral rights via an auction might not know the true market values of the natural resources present in his land, but might expect bids to be clustered around a high value or a low value based on the bidders' (who represent entities with domain expertise in, say, mining for rare materials) assessment of the resource quality. Then, based on a subset of revealed bids, the landowner can set informed reserve prices for other bidders to increase his revenue.

We then discuss an improvement to the VCG mechanism we call the {\em weakest-type VCG mechanism} (Section~\ref{sec:wcvcg}). While vanilla VCG charges an agent her externality measured relative to the welfare achievable by her non-participation, weakest-type VCG charges an agent her externality relative to the welfare achievable by the {\em weakest type} consistent with what the mechanism designer already knows about that agent. An agent's weakest type is the type from her type space that has the smallest impact on the efficient welfare. This idea is due to~\citet{krishna1998efficient}, who showed that a Bayesian version of the weakest-type mechanism is revenue optimal among all efficient, Bayes incentive compatible, and Bayes individually rational mechanisms (assuming the mechanism designer has access to the prior distribution from which agents' values are drawn). Our weakest-type VCG mechanism adapts Krishna and Perry's mechanism to a prior-free setting and allows for a more general model of agent type spaces, and we prove that it is revenue optimal subject to efficiency, (ex-post) incentive compatibility, and (ex-post) individual rationality.

We show how the payment scheme implemented by the weakest-type VCG mechanism lends a natural interpretation in terms of information rents for the agents. Indeed, vanilla VCG is equivalent to a pay-as-bid/first-price payment scheme with agent discounts equal to the welfare improvement they create for the system. Weakest-type VCG is equivalent to a pay-as-bid scheme with agent discounts equal to the welfare improvement they create for the system {\em over the welfare created by their weakest type}. That difference in welfare is precisely an agent's {\em information rent}: the less private information she holds to ``distinguish herself" from the weakest type, the smaller her discount.

\paragraph{Measuring prediction quality and weakest type computation} 

In Section~\ref{sec:quality} we devise an appropriate error measure to quantify the quality of a prediction that is intimately connected to the weakest-type mechanism described above---it is precisely the {\em predicted information rent}. We establish a number of properties of predictors under this error measure, illustrating that information rent---by reducing a prediction to a single number measuring the welfare of its weakest type---serves as a unifying way to analyze predictors, no matter how complex. One such property is that the ability of a prediction to boost revenue in our framework is not captured by some obvious measures of goodness. For example, one might expect that a good prediction should say something correct about an agent's true type---if it claims that an agent's type/valuation satisfies property $P$ we would expect that her true type ought to satisfy property $P$ for it to be a useful prediction. Rather counter-intuitively, even if the true type does not satisfy property $P$ our framework can gain meaningful revenue mileage from that prediction. However, we do demand from our mechanisms the following economic desideratum: whenever a prediction does say something correct about agents' true types, the mechanism ought to be efficient. We term this constraint {\em welfare consistency}. Welfare consistency is the natural relaxation of efficiency for mechanisms with predictors.

We also briefly discuss the computational complexity of finding weakest types (Section~\ref{sec:computation}). We show that weakest type computation can be formulated as a linear program with size equal to the size of the allocation space, which implies that weakest types can be found in time polynomial in the parameters defining the mechanism design environment. In environments where the size of the allocation space is prohibitively large ({\em e.g.,} in a combinatorial auction there are $(n+1)^m$ possible allocations of $n$ items to $m$ bidders and the seller), we show that the solution to the linear program can nonetheless be computed with polynomially many queries to an oracle for computing efficient allocations.

\paragraph{Prediction-augmented mechanisms and guarantees}
In Section~\ref{sec:main_guarantees} we present our main mechanism. It uses the information output by the predictors within the weakest-type VCG mechanism, but modifies that information via two tunable parameters. The parameters control a randomized relaxation of the prediction drawn from an exponential discretization of an agent's possible welfare levels.

We prove that our mechanism achieves strong welfare and revenue guarantees that are parameterized by errors in the predictions and the quality of the parameter tuning. Our benchmark for both welfare and revenue is the efficient welfare, denoted by $\OPT$. Our mechanism is {\em welfare consistent}, so whenever predictors are correct it attains the efficient welfare $\OPT$. When predictions are high quality, its revenue is competitive with $\OPT$. Its performance degrades gradually as the quality of predictions worsen.

{\em Prior-free efficient welfare $\OPT$, or total social surplus, is the strongest possible benchmark for both welfare and revenue.} We show that simplistic approaches that solely optimize for the {\em consistency} and {\em robustness} desiderata that have been studied in the mechanism design (and algorithm design more broadly) with predictions literature are brittle and overly sensitive to prediction errors (and it is in fact trivial to attain strong consistency and robustness bounds---hence our focus is on error-tolerant mechanisms). Our mechanism provides a more flexible, general, and robust alternative. Finally, we establish impossibility results on how much revenue a welfare consistent mechanism with predictors can extract. These results show that the performance of our main mechanism (which is welfare consistent) cannot be substantially improved.

\paragraph{Other forms of side information} In Section~\ref{sec:other_forms} we apply the weakest-type mechanism to three other formats of side information. First, in Section~\ref{sec:other_forms+uncertainty}, we describe a more general model of predictions that can express arbitrary degrees of uncertainty over an agent's type. Here, we generalize the main randomized mechanism described previously. Second, in Section~\ref{section:other_forms+subspace}, we derive new results in a setting where each agent's type is determined by a constant number of parameters. Specifically, agent types lie on constant-dimensional subspaces (of the potentially high-dimensional ambient type space) that are known to the mechanism designer (this is markedly different from our model of predictions and should be viewed as a restriction on the agent's type space itself). \emph{For example, a real-estate agent might infer a buyers' relative property values based on value per square foot.} When each agent's true type is known to lie in a particular $k$-dimensional subspace of the ambient type space, we modify the weakest-type mechanism by choosing weakest types randomly from a careful discretization of the subspace, to obtain revenue at least $\Omega(\OPT/k(\log H)^{k})$ while simultaneously guaranteeing welfare at least $\OPT/\log H$, where $H$ is an upper bound on any agent's value. Third, in Section~\ref{sec:other_forms+myerson}, we consider a textbook multidimensional mechanism design setup where the side information is in the form of a known prior distribution over agent types. Here, we show that the {\em revenue-optimal} (with no constraints other than incentive compatibility and individual rationality) Groves mechanism can be found by solving a separate single-parameter optimization problem for each agent. The optimization involves each agent's weakest type. Interestingly, our formulation recovers the optimal single-item auction of~\citet{Myerson81:Optimal} in a special case despite the fact that it is not globally revenue optimal in general.

\paragraph{Extension to affine-maximizer mechanisms} Finally, we show how the weakest-type framework can be extended beyond VCG to the class of affine-maximizer (AM) mechanisms~\citep{Roberts79:characterization}. We describe at least two attractive uses of a weakest-type AM in place of weakest-type VCG. First, it is well-known that in settings where agent types are drawn from a prior distribution, AMs can generate significantly more revenue than VCG. For example, work on {\em sample-based automated mechanism design}~\citep{sandholm2015automated,balcan2018general,curry2023differentiable} has shown that high-revenue AM parameters can be learned from data. Our techniques can then be appended as a post-processor to further improve the revenue of an already-tuned AM. Second, in many application domains it might make more sense to implement an allocation that maximizes a {\em weighted} version of welfare. For example, the mechanism designer might want to prioritize allocations that adequately reward small or minority-owned business (captured by multiplicative bidder weights). As another example, an auction designer might derive value from keeping some of the items for himself or from offering items for other uses outside the auction. In such cases, allocations that leave some items unsold might be prioritized if no bidder values those items competitively (captured by additive allocation boosts). We show how our main results extend to AMs, with appropriately redefined welfare and revenue benchmarks.

\subsection{Related work} We survey related work, discuss work subsequent to our initial conference publication~\citep{balcan2023bicriteria}, and briefly compare the present paper to our initial publication.

\paragraph{Side information in mechanism design} Various mechanism design settings have been studied under the assumption that some form of public side information is available. \citet{munoz2017revenue} study single-item (unlimited supply) single-bidder posted-price auctions with bid predictions. \citet{devanur2016sample} study the sample complexity of (single-parameter) auctions when the mechanism designer receives a distinguishing signal for each bidder. More generally, the active field of \emph{algorithms with predictions} aims to improve the quality of classical algorithms when machine-learning predictions about the solution are available~\citep{mitzenmacher2022algorithms}. There have been recent explicit connections of this paradigm to settings with strategic agents~\citep{agrawal2022learning, gkatzelis2022improved, balkanski2022strategyproof,banerjee2022online}. Most related to our work,~\citet{xu2022mechanism} study auctions for the sale of a (single copy of a) single item when the mechanism designer receives point predictions on the bidders' values. Unlike our approach, they focus on \emph{deterministic} modifications of a second-price auction. An important drawback of determinism is that revenue guarantees do not decay continuously as prediction quality degrades. For agents with values in $[1,H]$ there is an error threshold after which, in the worst case, only a $1/H$-fraction of revenue can be guaranteed (a vacuous guarantee not even competitive with VCG---in contrast we provide worst case robustness guarantees that are competitive with VCG, and even better in some cases). \citet{xu2022mechanism} prove that such a revenue drop is unavoidable by deterministic mechanisms. Finally, our setting is distinct from, but similar to in spirit, work that uses public attributes for market segmentation to improve revenue~\citep{balcan2005mechanism,balcan2020efficient}.

%\vspace{2mm}
\paragraph{Welfare-revenue tradeoffs in auctions} Welfare and revenue relationships in Bayesian auctions have been widely studied since the seminal work of~\citet{bulow1996auctions}. Welfare-revenue tradeoffs for second-price auctions with reserve prices in the single item setting have been quantified~\citep{hartline2009simple,daskalakis2011simple}, with some approximate understanding of the Pareto frontier~\citep{diakonikolas2012efficiency}. \citet{anshelevich2016pricing} study welfare-revenue tradeoffs in large markets,~\citet{aggarwal2009efficiency} study the efficiency of revenue-optimal mechanisms, and~\citet{abhishek2010efficiency} study the efficiency loss of revenue-optimal mechanisms.

\paragraph{Weakest types} The idea behind the weakest-type VCG mechanism is due to~\citet{krishna1998efficient} but the notion of a ``worst-off" type was studied before that in the context of bilateral (and more general) trade by~\citet{Myerson83:Efficient} and~\citet{cramton1987dissolving}. That study identified the worst-off type in a trading environment and used that characterization to characterize individually rational trading mechanisms---the connection to individual rationality is similar to the one we draw in Theorem~\ref{theorem:rev_optimal}.

\paragraph{Distribution-free agent type models} Our primary model of predictions is {\em distribution free} in that a prediction puts forward a postulate that an agent's true type belongs to some set, but conveys no further distributional information over that set. Such models of agent types have been previously studied, though in different contexts, by~\citet{hyafil2004regret} for minimax optimal automated mechanism design,~\citet{Holzman04:Bundling} to understand equilibria in combinatorial auctions, and~\citet{chiesa2015knightian} for agents who have uncertainty about their own private types.

\paragraph{Constant-parameter mechanism design} Revenue-optimal mechanism design for settings where each agent’s type space is of a constant dimension has been studied previously in certain specific settings. Single-parameter mechanism design is a well-studied topic dating back to the seminal work of~\citet{Myerson81:Optimal}, who (1) characterized the set of all truthful allocation rules and (2) derived the Bayesian optimal auction based on virtual values (a quantity that is highly dependent on knowledge of the agents' value distributions). \citet{kleinberg2013ratio} prove revenue guarantees for a variety of single-parameter settings that depend on distributional parameters. Financially constrained buyers with two-parameter valuations have also been studied~\citep{malakhov2009optimal,pai2014optimal}.

%\vspace{2mm}

\paragraph{Combinatorial auctions for limited supply} Our mechanism when agent types lie on known linear subspaces can be seen as a generalization of the well-known logarithmic revenue approximation that is achieved by a second-price auction with a random reserve price in the single-item setting~\citep{Goldberg01:Competitive}. Similar revenue approximations have been derived in multi-item settings for various classes of bidder valuation functions such as unit-demand~\citep{Guruswami05:profit}, additive~\citep{sandholm2015automated,Likhodedov05:Approximating,babaioff2020simple}, and subadditive~\citep{Balcan08:Item, chakraborty2013dynamic}. To the best of our knowledge, no previous techniques handle agent types on low-dimensional subspaces. Furthermore, our results are not restricted to combinatorial auctions unlike most previous research.

\paragraph{Sample-based mechanism design} Sandholm and Likhodedov~\citep{sandholm2015automated,Likhodedov04:Methods,Likhodedov05:Approximating} introduced the idea of automatically tuning the parameters of a mechanism based on samples of valuations. They also formulated mechanism design as a search problem within a parameterized family. They developed custom hill-climbing methods for learning various high-revenue auction formats, including affine maximizers, which we study in this paper. ~\citet{balcan2005mechanism} were the first to apply tools from machine learning theory to establish theoretical guarantees on the {\em sample complexity} of mechanism design. Since then, there has been a large body of work on both theoretical and practical aspects of data-driven mechanism design, for example, studying sample complexity of revenue maximization for various auction and mechanism families~\citep{morgenstern2016learning,cole2014sample,balcan2018general}, sample complexity of estimation of the degree of manipulability of a mechanism~\citep{balcan2019estimating,pieroth2024verifying}, and deep learning to design high-revenue mechanisms~\citep{dutting2019optimal,duan2024scalable,curry2023differentiable}. 

Beyond mechanism design, the field of {\em data-driven algorithm design}~\citep{balcan2020data} establishes theoretical foundations for a strand of what was largely empirical work on sample-based algorithm configuration (that predates sample-based mechanism design, {\em e.g.,}~\citet{Horvitz01:Bayesian}). This line of work is thematically similar to ours, but the main focus of our paper (and of the mechanism design with predictions literature more broadly) is to quantify the performance improvement obtainable based on the quality of predictions (that might have been obtained through learning from data). \citet{khodak2022learning} show how to ``learn the predictions in algorithms with predictions", thus concretely connecting the two research strands.

\paragraph{Work subsequent to our initial publication} After our initial conference publication~\citep{balcan2023bicriteria}, a number of papers have continued to grow the area of mechanism design with predictions. Still, this literature has largely focused on single-parameter settings. Specific applications include online single-item auctions~\citep{balkanski2024online}, single-parameter clock auctions~\citep{gkatzelis2025clock}, and randomized single-item auctions~\citep{caragiannis2024randomized}. (Unlike~\citet{caragiannis2024randomized}, we do not require knowledge of an upper bound on agent valuations. Their analysis and results hinges on a version of Myerson's lemma that does not generalize beyond single-parameter settings.)~\citet{lu2024competitive} use the idea of randomly modifying weakest types, that we proposed in the original conference version of the present paper~\citep{balcan2023bicriteria}, for error-tolerant design in a variety of single-parameter settings such as digital good auctions and auctions for the sale of multiple copies of a single item.~\citet{christodoulou2024mechanism} study a model where predictions are made directly on the mechanism's output. Specifically, the mechanism designer receives a prediction positing what the efficient allocation is, and the goal is to use that prediction to alleviate the computational cost of finding the true efficient allocation. We show that their model of ``output advice" is compatible with our model of predictors. Furthermore, we prove that qualitative predictions about the efficient allocation (like those studied by~\citet{christodoulou2024mechanism}) are too weak to improve revenue beyond VCG (the formal statement and proof are deferred to the appendix). Finally, \citet{balcan2025increasing} show how to learn the welfare level of weakest types from data,~\citet{prasad2026weakest} design new type-space-dependent core-selecting auctions and develop constraint generation algorithms for computing weakest types, and~\citet{prasad2025revenue} derive the revenue-optimal efficient mechanism when type spaces are completely general (the construction is a generalization of weakest-type VCG).

\paragraph{Comparison to our initial publication} The key differences between the present work and our initial conference publication~\citep{balcan2023bicriteria} are as follows: (i) our notion of prediction error is simpler (it only depends on the predicted welfare created, while the old notion of error depended on the geometry of the prediction) and provides a complete characterization of predictions that achieve a given payment (the old notion of error would have deemed some predictors unusable, but our notion of error can use such predictors to extract meaningful payments); (ii) our main mechanism is stronger---in particular it is always at least as good as the old version and in most cases it is much better---and our analysis yields stronger revenue guarantees; (iii) we provide
impossibility results that shed light on barriers to our main mechanisms; (iv) we apply the weakest-type idea to formulate the revenue-optimal Groves mechanism in multidimensional settings, and show that in a special case of the single-item setting our formulation recovers Myerson's optimal auction~\citep{Myerson81:Optimal} (this contribution is new and was not present in our initial publication); and (v) we include an extended discussion of weakest-type affine maximizer mechanisms, including an extension of the revenue-optimality of weakest-type VCG subject to efficiency to the affine maximizer case (our initial publication contained a definition of a weakest-type affine maximizer mechanism, but the revenue optimality result is new and a contribution of the present work).

\section{Problem formulation, example applications, and weakest-type VCG}
\label{section:problem_formulation}

We consider a general multidimensional mechanism design setting with a finite allocation space $\Gamma$ and $n$ agents. $\Theta_i$ is the ambient type space of agent $i$. Agent $i$'s true private type $\theta_i\in\Theta_i$ determines her value $v(\theta_i, \alpha)$ for allocation $\alpha\in\Gamma$. We will interpret $\Theta_i$ as a subset of $\R^{\Gamma}$, so $\theta_i[\alpha] = v(\theta_i, \alpha)$. We use $\btheta\in\bigtimes_{i=1}^n\Theta_i$ to denote a profile of types and $\btheta_{-i}\in\Theta_{-i}\coloneq\bigtimes_{j\neq i}\Theta_i$ to denote a profile of types excluding agent $i$. We now introduce our main model of side information which is the focus of Sections~\ref{sec:quality} and~\ref{sec:main_mechanism} (we discuss other models in Section~\ref{sec:other_forms}). First, the {\em ambient type space} $\Theta_i$ is assumed to convey no information about each agent, that is, $\Theta_i = \R_{\ge 0}^{\Gamma}$ (this is the standard assumption in the mechanism design literature). The mechanism designer has access to a set-valued {\em predictor} $T_i:\bigtimes_{j\neq i}\Theta_j\to\cP(\Theta_i)$ for each agent. Predictor $T_i$ takes as input {\em the revealed types $\btheta_{-i}$ of all agents excluding $i$} and outputs a set $T_i(\btheta_{-i})\subseteq\Theta_i$ that represents a prediction that agent $i$'s true type lies in $T_i(\btheta_{-i})$. The mechanism designer has no apriori guarantees about the quality (or validity) of the prediction output by $T_i$. {\em We emphasize that a prediction about agent $i$ can depend on the revealed types of all other agents.} This begets a rich and expressive language of prediction that can incorporate, for example, market insights by means of analyzing other agents' true types (a form of learning within an instance~\citep{baliga2003market,balcan2005mechanism,balcan2021learning}), interlinked knowledge about multiple agents' types, {\em etc.} Importantly---as we will later demonstrate---allowing predictions to carry such a great level of expressive power poses no barrier to incentive compatibility of our mechanisms. Finally, we point out that the modeling assumption of a separate predictor for each agent does not in any way restrict the ability of side information to convey relationships between the types of different agents: indeed, the claim that ``agent types $(\theta_1,\ldots,\theta_n)$ satisfy property $P$" induces predictors $T_1,\ldots, T_n$ where $T_i$ is defined by $T_i(\btheta_{-i}) = \{\hat{\theta}_i : (\hat{\theta_i}, \btheta_{-i})\text{ satisfies } P\}$.

A mechanism with predictors is specified by an allocation rule $\alpha(\btheta; T_1,\ldots,T_n)\in\Gamma$ and a payment rule $p_i(\btheta; T_1,\ldots, T_n)\in\R$ for each agent $i$. We assume agents have quasilinear utilities. A mechanism is \emph{incentive compatible (IC)} if $\theta_i\in\argmax_{\theta_i'\in\Theta_i} \theta_i[\alpha(\theta_i',\btheta_{-i};T_1,\ldots, T_n)] - p_i(\theta_i',\btheta_{-i};T_1,\ldots,T_n)$ holds for all $i, \theta_i\in\Theta_i, \btheta_{-i}\in\Theta_{-i}$, that is, agents are incentivized to report their true type regardless of what other agents report %and regardless of the predictors used by the mechanism 
(this definition is equivalent to the usual notion of dominant-strategy IC and simply stipulates that predictors ought to be used in an IC manner). A mechanism is \emph{individually rational (IR)} if $\theta_i[\alpha(\theta_i,\btheta_{-i};T_1,\ldots,T_n)] -p_i(\theta_i,\btheta_{-i};T_1,\ldots,T_n)\ge 0$ holds for all $i,\theta_i,\btheta_{-i}$. We will analyze a variety of randomized mechanisms that randomize over IC and IR mechanisms. Such randomized mechanisms are thus IC and IR in the strongest possible sense (as supposed to Bayes-IC/IR which is weaker and holds only in expectation). {\em An important note: no assumptions are made on the veracity of $T_i$, and agent $i$'s misreporting space is the full ambient type space $\Theta_i$.}

Given reported types $\btheta$, define $w(\btheta) = \max_{\alpha\in\Gamma}\sum_{i=1}^n\theta_i[\alpha]$ to be the {\em efficient} social welfare. Let $\alpha^* = \alpha^*(\btheta) =\argmax_{\alpha}\sum_{i=1}^n\theta_i[\alpha]$ denote the {\em efficient} allocation, that is, the allocation that achieves $w(\btheta)$. Our benchmark for welfare and revenue is the prior-free efficient welfare $\OPT = w(\btheta)$, which is the strongest possible benchmark for both welfare and revenue.

Before presenting applications, we conclude the setup with an economic desideratum---{\em welfare consistency}---we require of our mechanisms. It is the natural relaxation of efficiency for a mechanism with predictors. Welfare consistency demands that the mechanism with predictors be efficient if the predictors say something that is correct.
\begin{definition}[Welfare consistency]
    A mechanism with predictors is {\em welfare consistent} if, whenever $\theta_i\in T_i(\btheta_{-i})$ for all $i$, it implements the efficient allocation and obtains the efficient welfare $\OPT$.
\end{definition}

\subsection{Example applications}
\label{subsection:examples}

Our model of side information within the rich language of multidimensional mechanism design allows us to capture a variety of different problem scenarios where both welfare and revenue are desired objectives. We list a few examples of different multidimensional mechanism settings along with examples of different varieties of predictions.

\begin{itemize}
    \item Combinatorial auctions: There are $m$ indivisible items to be allocated among $n$ bidders (or to no one). The allocation space $\Gamma$ is the set of $(n+1)^m$ allocations of the items and $\theta_i[\alpha]$ is bidder $i$'s value for the bundle of items she is allocated by $\alpha$. Let $X$ and $Y$ denote two of the items for sale. The predictor $T_i(\btheta_{-i}) = \{\theta_i : \theta_i[XY]\ge 3/2\cdot \theta_j[XY], \theta_i[XY]\ge \theta_i[X] + \theta_i[Y]\}$ represents the prediction that bidder $i$'s value for the bundle $XY$ is at least $50\%$ greater than bidder $j$'s value for the same bundle and that items $X$ and $Y$ are complements for her. Here, $T_i(\btheta_{-i})$ is the intersection of linear constraints. 
    \item Matching markets: There are $m$ items (\emph{e.g.}, houses) to be matched to $n$ buyers. The allocation space $\Gamma$ is the set of matchings on the bipartite graph $K_{m,n}$ and $\theta_i[\alpha]$ is buyer $i$'s value for the item $\alpha$ assigns her. Let $\alpha_1,\alpha_2,\alpha_3$ denote three matchings that match house 1, house 2, and house 3 to agent $i$, respectively. The {\em type space restriction} $\Theta_i = \{\theta_i : \theta_i[\alpha_1] = 2\cdot\theta_i[\alpha_2] = 0.75\cdot\theta_i[\alpha_3]\}$ represents the information that agent $i$ values house 1 twice as much as house 2, and $3/4$ as much as house 3. Here, $\Theta_i$ is the linear space given by $\operatorname{span}(\langle 1, 1/2, 4/3\rangle)$ (this model is further studied in Section~\ref{section:other_forms+subspace} and differs from our main prediction model). 
    \item Fundraising for a common amenity: A multi-story office building that houses several companies is opening a new cafeteria on a to-be-determined floor and is raising construction funds. The allocation space $\Gamma$ is the set of floors of the building and $\theta_i[\alpha]$ is the (inverse of the) cost incurred by building-occupant $i$ for traveling to floor $\alpha$. The set $T_i = \{\theta_i : \lVert\theta_i - \theta_i^*\rVert_p\le k\}$ postulates that $i$'s true type is no more than $k$ away from $\theta_i^*$ in $\ell_{p}$-distance, which might be derived from an estimate of the range of floors agent $i$ works on based on the company agent $i$ represents. Here, $T_i$ is given by a (potentially nonlinear) distance constraint and has no dependence on the other agents' revealed types. 
    \item Bidding for a shared outcome: A delivery service that offers multiple delivery rates (priced proportionally) needs to decide on a delivery route to serve $n$ customers. The allocation space $\Gamma$ is the set of feasible routes and $\theta_i[\alpha]$ is agent $i$'s value for receiving her packages after the driving delay specified by $\alpha$. Let $\alpha_t$ denote an allocation that imposes a driving delay of $t$ on agent $i$. The set $T_i(\btheta_{-i}) = \{\theta_i : \theta_i[\alpha_0]\ge\max_{j\neq i}\theta_j[\alpha_0] + \$50, \theta_i[\alpha_{t+1}]\ge f_t(\theta_i[\alpha_{t}])\;\forall t\}$ is the prediction that agent $i$ is willing to pay \$50 more than any other agent to receive her package as soon as possible, and is at worst a time discounter determined by (potentially nonlinear) discount functions $f_t$. Here, the complexity of $T_i$ is determined by the function $f_t$. 
\end{itemize}

%In our results we will assume that $\Theta = [0, H]^{\Gamma}$ imposing a cap $H$ on any agent's value for any allocation. This is the only problem-specific parameter in our results. In the above four bulleted examples $H$ represents the maximum value any agent has for the grand bundle of items, any available house, the cafeteria opening on her floor, and receiving her packages with no delay, respectively. {\color{blue}TODO:}

\subsection{The weakest-type VCG mechanism}\label{sec:wcvcg} The classical Vickrey-Clarke-Groves (VCG) mechanism~\citep{Vickrey61:Counterspeculation,Clarke71:Multipart,Groves73:Incentives} elicits agent types $\btheta$, implements the efficient allocation $\alpha^*$ and charges agent $i$ a payment $p_i(\btheta) = w(0,\btheta_{-i}) - \sum_{j\neq i}\theta_j[\alpha^*]$ which is agent $i$'s externality. VCG is generally highly suboptimal when it comes to revenue~\citep{Ausubel06:Lovely,varian2014vcg} (and conversely mechanisms that shoot for high revenue can be highly welfare suboptimal). However, if efficiency is enforced as a constraint of the mechanism design (in addition to IC and IR), then the following \emph{weakest-type VCG mechanism}, first introduced by~\citet{krishna1998efficient} in a Bayesian form, is in fact revenue optimal (Krishna and Perry call it the generalized VCG mechanism). While VCG payments are based on participation externalities, weakest-type VCG payments are based on agents being replaced by their \emph{weakest types} who have the smallest impact on welfare. This approach generates strictly more revenue than vanilla VCG.~\citet{krishna1998efficient} proved that the Bayesian version of weakest-type VCG is revenue optimal among all efficient, Bayes IC, and Bayes IR mechanisms. We next present the weakest-type VCG mechanism, which is a generalization of their mechanism in a prior-free setting.

To describe our weakest-type VCG mechanism, we depart slightly from the language of predictors and focus on what is already known to the mechanism designer as conveyed by the {\em joint type space} of the agents. That is, there is a joint type space $\vec{\Theta}\subseteq\bigtimes_{i=1}^n\R_{\ge 0}^{\Gamma}$ and the mechanism designer knows that the revealed type profile $\btheta = (\theta_1,\ldots,\theta_n)$ belongs to $\bTheta$. The assumption of a joint type space with no further restrictions allows for a rich knowledge structure that can capture arbitrary relationships between agents; given revealed types $\btheta_{-i}$ for all agents excluding $i$, the mechanism designer knows that the revealed type of agent $i$ belongs to the set $\{\hat{\theta}_i : (\hat{\theta}_i,\btheta_{-i})\in \vec{\Theta}\}$ (which is a {\em projection} of the joint type space $\vec{\Theta}$ onto agent $i$'s space, given $\btheta_{-i}$). We define the {\em weakest-type VCG mechanism} in this general setting. 

\begin{figure}
\begin{tcolorbox}[colback=white, parbox=false]
\underline{Weakest-type VCG}\\
Input: joint type space $\bTheta\subseteq\bigtimes_{i=1}^n\R^{\Gamma}_{\ge 0}$.

\begin{itemize}
    \item Agents asked to reveal types $\btheta=(\theta_1,\ldots,\theta_n)$.
    \item The efficient allocation $\alpha^* = \argmax_{\alpha\in\Gamma}\sum_{i=1}^n\theta_i[\alpha]$ is implemented and agent $i$ is charged a payment of \begin{equation}\label{eq:wt}p_i(\btheta) = \min_{\ttheta_i:(\ttheta_i,\btheta_{-i})\in\bTheta}w(\ttheta_i,\btheta_{-i}) - \sum_{j\neq i}\theta_j[\alpha^*].\end{equation}
\end{itemize}
\end{tcolorbox}
\end{figure}
We call $\ttheta_i$ in the minimization~\eqref{eq:wt} agent $i$'s {\em weakest type}.\footnote{The weakest-type VCG mechanism is {\em not} equivalent to a mechanism that actually adds ``fake" competitors into the mechanism environment (as is the case with mechanisms that use ``phantom bidders"~\citep{Sandholm13:Very}). Indeed, consider the mechanism that, for each agent $i$, adds an agent $n+i$ who is the weakest type and then runs vanilla VCG over the augmented set of $2n$ agents. That mechanism is not necessarily efficient and is thus not equivalent to weakest-type VCG. As an example consider a two-item, two-bidder auction where each bidder only wants one item. Say $v_1(X) = v_1(Y) = 10$, $v_2(X) = 5, v_2(Y) = 4$. The efficient allocation gives $Y$ to $1$ and $X$ to $2$ and has welfare $15$. Suppose agent $1$'s type space is $\{v_1 : v_1(X) \ge 7\}$ and agent $2$'s type space is $\R_{\ge 0}^2$ (so type spaces are independent). So agent $1$'s weakest type is $\widetilde{v}_1(X) = 7, \widetilde{v}_1(Y) = 0$, and agent $2$'s weakest type is $\widetilde{v}_2(X) = \widetilde{v}_2(Y) = 0$. Running vanilla VCG among $\{v_1, \widetilde{v}_1, v_2, \widetilde{v}_2\}$ yields an efficient allocation that gives $X$ to $\widetilde{1}$ and $Y$ to $1$. So, while an agent's weakest type will never take that agent's items away from them, they can win some other agent's items. In weakest-type VCG, the weakest types are not real agents and only serve to boost prices---they do not receive any utility from the allocation.} If $(0, \btheta_{-i})\in\bTheta$, $\ttheta_i = 0$, and $p_i(\btheta) = w(0,\btheta_{-i}) - \sum_{j\neq i}\theta_j[\alpha^*]$ is the vanilla VCG payment. On the other extreme if $\{\hat{\theta}_i : (\hat{\theta}_i,\btheta_{-i})\in\bTheta\} = \{\theta_i\}$, that is, the the joint type space exactly conveys agent $i$'s true type, $p_i(\btheta) = \theta_i[\alpha^*]$ and agent $i$'s total value is extracted as payment. \citet{krishna1998efficient} essentially prove a weaker form (they assume independent type spaces, that is, $\bTheta = \bigtimes_i \Theta_i$ has a product structure) of the following result in the Bayesian setting (that is, they only require the weaker constraints of Bayes IC and Bayes IR). We reproduce that here in a prior-free setting with our more general model of type spaces. We present two proofs: one based on the revenue equivalence theorem (stated subsequently) and one based on the result of~\citet{Holmstroem79:Groves} concerning Groves mechanisms.

\begin{theorem}[Revenue equivalence]\label{theorem:revenue_equivalence}
  Let $\bTheta$ be a convex joint type space. Let $p$ be an IC pricing rule implementing the efficient allocation and let $p'$ be any other pricing rule. Then, $p'$ is IC if and only if, for each $i$, there exists $h_i(\btheta_{-i})$ such that $p_i'(\btheta) = p_i(\btheta) + h_i(\btheta_{-i})$.
\end{theorem}

Revenue equivalence is a celebrated result from economics dating back to~\citet{Green77:Characterization,Holmstroem79:Groves,Myerson81:Optimal}. Our presentation is based on~\citet[Theorem 4.3.1]{vohra2011mechanism}.

\begin{theorem}\label{theorem:rev_optimal}
Let $\bTheta$ be a convex joint type space. The weakest-type VCG mechanism is incentive compatible and individually rational. Furthermore, it maximizes revenue among all efficient, incentive compatible, and individually rational mechanisms.
\end{theorem}

\begin{proof}
    Weakest-type VCG is incentive compatible since it is a Groves mechanism~\citep{Groves73:Incentives}, that is, $\min_{\ttheta_i : (\ttheta_i, \btheta_{-i})\in\bTheta} w(\ttheta_i,\btheta_{-i})$ (called the {\em pivot} term) has no dependence on agent $i$'s revealed type $\theta_i$. Furthermore, agent $i$'s utility is $\sum_{j=1}^n \theta_j[\alpha^*] - \min_{\ttheta_i : (\ttheta_i, \btheta_{-i})\in\bTheta}(\max_{\alpha}\sum_{j\neq i} \theta_j[\alpha] + \ttheta_i[\alpha])\ge \sum_{j=1}^n \theta_j[\alpha^*] - \max_{\alpha} \sum_{j=1}^n \theta_j[\alpha] = 0,$ which proves individual rationality. The proof that weakest-type VCG is revenue optimal follows from the revenue equivalence theorem. Let $p_i(\btheta)$ be the weakest-type VCG payment rule, and let $p_i'(\btheta)$ be any other payment rule that also implements the efficient allocation rule. By revenue equivalence, for each $i$, there exists $h_i(\btheta_{-i})$ such that $p_i'(\theta_i, \btheta_{-i}) = p_i(\theta_i, \btheta_{-i}) + h_i(\btheta_{-i})$. Suppose $\btheta$ is a profile of types such that $p_i'$ generates strictly greater revenue than $p_i$, that is, $\sum_{i=1}^n p_i'(\btheta) > \sum_{i=1}^n p_i(\btheta)$. Equivalently $\sum_{i=1}^n p_i(\theta, \btheta_{-i}) + h_i(\btheta_{-i}) > \sum_{i=1}^n p_i(\theta_i, \btheta_{-i})$. Thus, there exists $i^*$ such that $h_{i^*}(\btheta_{-i^*}) > 0$. Let $$\ttheta_{i^*}=\argmin_{\theta_{i^*}' : (\theta'_{i^*}, \btheta_{-i^*})\in\bTheta}w(\theta'_{i^*}, \btheta_{-i^*})$$ be the \emph{weakest type} with respect to $\btheta_{-i^{*}}$. If weakest-type VCG is run on the type profile $(\ttheta_{i^*}, \btheta_{-i^*})$, the agent with type $\ttheta_{i^*}$ pays their value for the efficient allocation. In other words, the individual rationality constraint is binding for $\ttheta_{i^*}$. Since $h_{i^*}(\btheta_{-i^*}) > 0$, $p_i'$ violates individual rationality, which completes the proof. 
    
    We provide an alternate proof of the revenue optimality of weakest-type VCG via the result of~\citet{Holmstroem79:Groves} that, for convex type spaces, any efficient and IC mechanism is a Groves mechanism. (Technically,~\citet{Holmstroem79:Groves} assumes that agent type spaces are independent, that is $\bTheta$ has a product structure. We can circumvent this issue as follows. Let $\Theta_i(\btheta_{-i}) = \{\hat{\theta}_i : (\hat{\theta}_i, \btheta_{-i})\in\bTheta\}$. We can then give $\bTheta$ product structure by writing $\bTheta = \Theta_1(\btheta_{-1})\times\cdots \times\Theta_n(\btheta_{-n})$. Each $\Theta_i(\btheta_{-i})$ is convex since $\bTheta$ itself is convex.) Then, the revenue-maximizing Groves payment scheme $h_i(\btheta_{-i})$ solves for each agent $i$ $$\max \left\{h_i(\btheta_{-i}): \ttheta_i[\alpha^*] - \left(h_i(\btheta_{-i}) - \sum_{j\neq i}\theta_j[\alpha^*]\right)\ge 0\;\;\;\;\forall\ttheta_i\text{ s.t. } (\ttheta_i,\btheta_{-i})\in \bTheta\right\},$$ which yields the maximum Groves payment subject to IR constraints for all possible types $\ttheta_i$ such that $(\ttheta_i,\btheta_{-i})\in\bTheta$ (and has no dependence on agent $i$'s revealed type $\theta_i$). Rewrite the constraint as $h_i(\btheta_{-i})\le w(\ttheta_i,\btheta_{-i})~\forall~\ttheta_i\text{ s.t. } (\ttheta_i,\btheta_{-i})\in\bTheta$. So, $h_i(\btheta_{-i})\le\min_{\ttheta_i: (\ttheta_i,\btheta_{-i})\in\bTheta}w(\ttheta_i,\btheta_{-i})$ which is precisely weakest-type VCG. 
\end{proof}

Vanilla VCG payments can be interpreted as implementing a first-price/pay-as-bid scheme with a discount for each agent equal to the welfare improvement they create by participating. Weakest-type VCG payments implement a pay-as-bid scheme with a discount equal to the welfare improvement an agent creates over the {\em weakest type} in their type space. The discount can be directly interpreted as an {\em information rent} ({\em e.g.,}~\citet{borgers2015introduction}) incurred by the agents. The more private information an agent has (measured by how far away her IR constraint in weakest-type VCG is from binding---we discuss this aspect formally in the subsequent section) relative to the information conveyed by the type space, the greater her discount. Said another way, the more private information an agent has to distinguish herself from the weakest type, the greater her discount. The weakest type is also the {\em most reluctant} type~\citep{krishna1998efficient} in that it receives zero utility from participation. Finally, weakest types can be viewed as a sophisticated form of reserve pricing (and any reserve pricing structure atop VCG can be expressed as a weakest-type mechanism).

Recent work by~\citet{prasad2025revenue} extends Theorem~\ref{theorem:rev_optimal} to completely general type spaces (for example, disconnected sets representing disjunctive knowledge about agent types). It is shown that the revenue-optimal efficient mechanism belongs to a larger class of ``component-wise" Groves mechanisms; the appropriate payment-maximizing {\em pivot} terms can then be computed via a shortest-paths formulation that arises from the underlying network flow structure to the incentives. (\citet{kinoshita2024socially} derive a similar result for the special case of finite type spaces.)

We conclude our discussion of weakest-type VCG with the observation that the weakest-type VCG mechanism can loosely be interpreted as a prior-free analogue of the seminal Bayesian total-surplus-extraction mechanism of~\citet{cremer1988full} for correlated agents (generalized to infinite type spaces by~\citet{mcafee1992correlated}). This is an interesting connection to explore further in future research.

%What if predictors are not known to be valid (so weakest-type VCG is not necessarily IR), and even if they are valid, how much revenue does weakest-type VCG leave on the table? 

In the following section we return to our model of predictors, which can now be interpreted as conveying the same kind of knowledge as a joint type space, but without the guarantee that the information conveyed about the agents is actually correct. That guarantee---which equivalently stipulates that agent $i$'s {\em misreporting space} (that determines her IC and IR constraints) is restricted to the set $\{\hat{\theta}_i : (\hat{\theta}_i, \btheta_{-i})\in \bTheta\}$---was required in Theorem~\ref{theorem:rev_optimal} to ensure individual rationality of weakest-type VCG. Our setting of mechanism design with predictors can thus be interpreted as removing this constraint on the misreporting space: predictors $\{T_i\}$ have no apriori guarantees on their veracity and we must use them in a way that respects IC and IR constraints for a misreporting space equal to the entire ambient type space. 

%Our meta-mechanism (Section~\ref{section:meta_mechanism}) is a generalization of WCVCG that uses side information sets rather than the ambient type space to determine payments. (Misreporting is not limited to side information sets.) Our meta-mechanism relaxes efficiency in order to use the side information to boost revenue.

\section{Measuring predictor quality via the welfare of a weakest type}\label{sec:quality}

Vanilla VCG payment can be expressed as $p_i^{\VCG} = \theta_i[\alpha^*] - (w(\theta_i,\btheta_{-i}) - w(0, \btheta_{-i}))$, and admits a natural interpretation as a pay-as-bid term ($\theta_i[\alpha^*]$) minus a discount equal to agent $i$'s information rent ($w(\theta_i,\btheta_{-i}) - w(0, \btheta_{-i})$). More generally, weakest type payment $p_i^{\WT} = \theta_i[\alpha^*] - (w(\theta_i,\btheta_{-i}) - \min_{\ttheta_i\in\Theta_i(\btheta_{-i})}w(\ttheta_i, \btheta_{-i}))$ admits the same interpretation. Information rent measures precisely the value of agent $i$'s private information over what the mechanism designer already knows prior to type revelation---that agent $i$ creates welfare at least $\min_{\ttheta_i\in\Theta_i(\btheta_{-i})} w(\ttheta_i,\btheta_{-i})$ for the system.

Our measure of predictor error is the negation of the predictor's purported information rent: let $$\err_i = \err_i(\btheta_{-i}) = \min_{\ttheta_i\in T_i(\btheta_{-i})} w(\ttheta_i,\btheta_{-i})-  w(\theta_i, \btheta_{-i}).$$ If $\err_i \le 0$, the predictor underestimates the welfare agent $i$ contributes to the system; if $\err_i > 0$ it overestimates welfare. By measuring a predictor's quality by its purported information rent, we provide a unified way of analyzing an extremely general class of expressive predictions through their information rents. The purported information rent of the general statement ``agent $i$'s type satisfies property $P(\btheta_{-i})$" is precisely $-\err_i = \min_{\ttheta_i \text{ satisfies } P(\btheta_{-i})} w(\ttheta_i,\btheta_{-i}) - w(\theta_i,\btheta_{-i})$. So, no matter how complex property $P$ might be, information rent is concerned with a single number: the predicted welfare created by the predicted weakest type.

\subsection{Properties of predictors} The following result is immediate from the definition of $\err_i$ and Theorem~\ref{theorem:rev_optimal}.

\begin{theorem}\label{theorem:revenue_characterization}
    Let $T_1,\ldots, T_n$ be predictors such that $\err_i\le 0$ for all agents $i$. The mechanism that implements the efficient allocation and prices given by $p_i = \min_{\ttheta_i\in T_i(\btheta_{-i})}w(\ttheta_i,\btheta_{-i}) - \sum_{j\neq i}\theta_j[\alpha^*]$ is IC, IR, welfare consistent, and extracts payment $p_i = \theta_i[\alpha^*] - |\err_i|$ from each agent---and thus generates revenue equal to $\OPT - \sum_i |\err_i|.$
\end{theorem}

We describe an equivalent geometric characterization of predictors that generate a certain payment in the framework above. Let $L_w(\btheta_{-i}) = \left\{\ttheta_i : w(\ttheta_i,\btheta_{-i})=w\right\}$ be a {\em welfare level set} and let $L_{\ge k}(\btheta_{-i}) = \left\{\ttheta_i : w(\ttheta_i,\btheta_{-i})\ge w\right\}$. The following propositions follow immediately from the definitions. 

\begin{proposition}
    $L_{\ge w}(\btheta_{-i})$ is the union of axis-parallel halfspaces.
\end{proposition}
\begin{proof} We have
    $$L_{\ge w}(\btheta_{-i}) = \left\{\ttheta_i : \bigvee_{\alpha\in\Gamma}\ttheta_i[\alpha]\ge w - \sum_{j\neq i}\theta_j[\alpha]\right\} = \bigcup_{\alpha\in\Gamma}\left\{\ttheta_i : \ttheta_i[\alpha]\ge w - \sum_{j\neq i}\theta_j[\alpha]\right\}.$$ 
\end{proof}

\begin{proposition}
    A predictor $T_i$ has error $\err_i$ if and only if it intersects the level set $L_{w(\theta_i,\btheta_{-i})+\err_i}$ and does not intersect any $L_w$ with $w < w(\theta_i,\btheta_{-i})$. Predictor $T_i$ has $\err_i\le 0$ if and only if $L_{w(\theta_i,\btheta_{-i})}(\btheta_{-i})\subseteq L_{w(\ttheta_i,\btheta_{-i})}(\btheta_{-i})$ where $\ttheta_i$ is the weakest type in $T_i(\btheta_{-i})$.
\end{proposition}

\begin{remark}
An important consequence is that if $T_i(\btheta_{-i})$ intersects $L_{w(0,\btheta_{-i})}$, $T_i$ is a ``useless" predictor in that the information it provides is not even strong enough to say that agent $i$ generates any extra welfare to the system. More precisely, $\min_{\ttheta_i\in T_i(\btheta_{-i})}w(\ttheta_i,\btheta_{-i}) = w(0,\btheta_{-i})$, so the predictor's purported information rent is equal to that of VCG. In the appendix, we show that a general class of predictors that only make {\em qualitative} predictions about the efficient allocation suffer from this issue.~\citet{christodoulou2024mechanism}, in work subsequent to our initial conference publication, use such predictions to enhance polynomial-time welfare-maximization algorithms. So, while their model of predictions is indeed compatible and subsumed by ours, the predictions they study are too weak to generate meaningful revenue in our setting (Proposition~\ref{prop:qualitative} in Appendix A). Interestingly, we are able to obtain meaningful revenue gains from predictors that do not even contain the true type, as illustrated by the following example.
\end{remark}

\begin{example}
    Consider an allocation space $\Gamma = \{\alpha_1,\alpha_2\}$ with two outcomes as depicted in Figure~\ref{fig:predictions}. Suppose $\btheta_{-i}$ is such that $\sum_{j\neq i}\theta_j[\alpha_1] = 1$ and $\sum_{j\neq i}\theta_j[\alpha_2] = 3$, so $w(0, \btheta_{-i}) = 3$. Welfare level sets are depicted by the bold black lines. $L_{w(0,\btheta_{-i})}$ is the set of types $\ttheta_i$ that do not create any additional welfare for the system which is precisely the set $\{\ttheta_i : \ttheta[\alpha_1]\le 2\wedge\ttheta[\alpha_2] = 0\}$. The sets enclosed by the dotted lines represent the outputs of two different predictors. Each set's weakest type generates a welfare of $w(\ttheta_i,\btheta_{-i}) = 4$ as each prediction intersects $L_4$ and no lower level set. The true type of agent $i$, $\theta_i$, generates welfare $w(\theta_i,\btheta_{-i}) = 5$. Therefore, both predictors have the same error of $\err_i = -1$ and their weakest types generate the same welfare of $4$. This is in spite of the fact that only one of the predictions actually contains the true type (the polygon) while the other (the ellipse) is quite far away from the true type.
\end{example}

\begin{figure}[t]
    \centering
    \includegraphics[width=0.6\linewidth, trim={37.5cm 24.4cm 0cm, 0cm},clip  ]{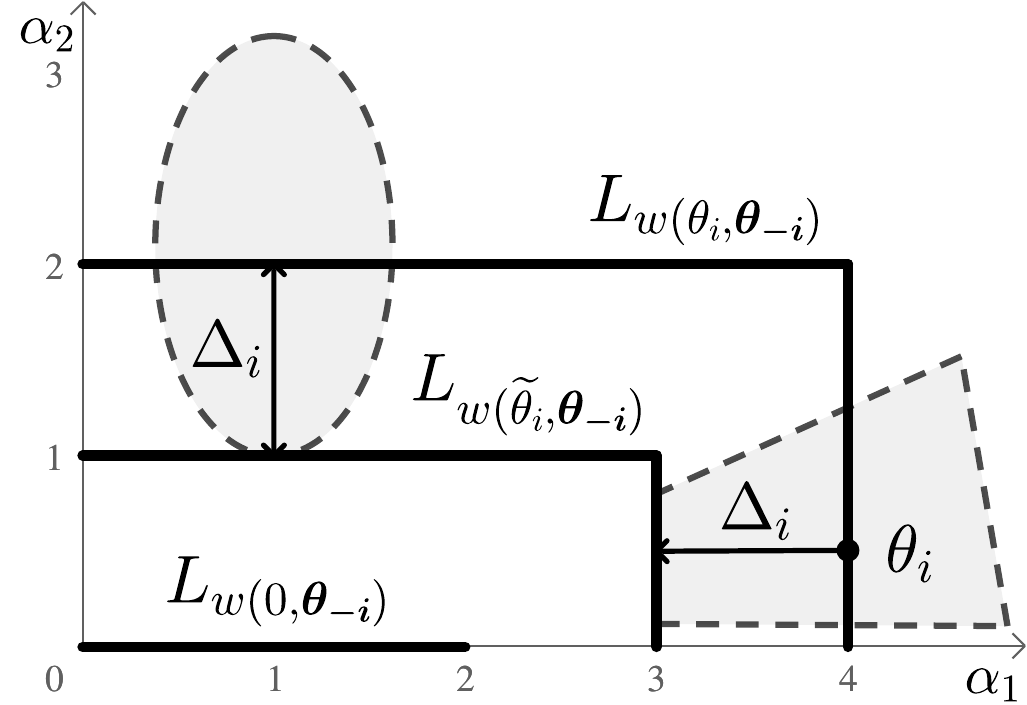}
    \caption{Two different predictions (the ellipse and polygon displayed with dashed boundaries) that are equivalent in the sense that their weakest types create the same amount of welfare $w(\ttheta_i,\btheta_{-i}) = w(\theta_i,\btheta_{-i}) - \Delta_i$ for the system and thus generate the same weakest-type payments for agent $i$, despite the fact that one prediction (the polygon) contains the true type and the other (the ellipse) completely misses the true type. Welfare level sets are depicted by the solid black lines.}
    \label{fig:predictions}
\end{figure}

\begin{remark}
An important property of this framework highlighted by the above example is that predictors need not know anything about the precise structure of types. As long as they provide a reasonable (under)estimate of agent $i$'s competitive level given $\btheta_{-i}$ they can be used in the weakest-type mechanism to fruitfully boost revenue over vanilla VCG. Consider the setting of a combinatorial auction with $m$ items and $n$ bidders wherein a bidder's ambient valuation space is $\R_{\ge 0}^{2^m}$ since she can express a value for any bundle of items. In practice, the auction designer might set a cap on the number of distinct bundles any single bidder can submit bids for (as was done in the spectrum auctions held in the UK and Canada~\citep{ausubel2017practical}). In this case, a predictor need not reckon with the question of what bundles the bidder will bid on (which is ostensibly a very difficult prediction problem that would require unrealistic insight into a bidder's bidding strategy). It only needs to provide a good estimate of her competitive level as measured by the welfare she creates for the system, which we posit is a much more reasonable and practically plausible prediction task. Indeed, in high-stakes auctions, the auction designer might reasonably expect a certain bidder (who, say, represents a large conglomerate as is the case in sourcing and spectrum auctions) to win at least some items. That already provides the auction designer sufficient knowledge to implement a weakest-type pricing scheme that extracts greater revenues than VCG. To summarize, a predictor like the ellipse in Figure~\ref{fig:predictions} that is completely wrong about the bundles bidder $i$ bids on is still perfectly good in our framework, and therefore predictors have substantial flexibility in the information they convey.\end{remark}

%Predictions thus have a good deal of leeway in the above sense. A predictor can be wildly inaccurate about the ``structure" of an agent's type, but as long as it accurately predicts the welfare she creates it fruitfully yields increased payments. 
Another way in which predictors have leeway is the fact that only the prediction error $\err_i(\btheta_{-i})$ given the revealed types of the other agents affects payments. Predictor $T_i$ could be wildly inaccurate on a different set of types, that is, $|\err_i(\btheta_{-i}')|$ might be large, but as long as $\err_i(\btheta_{-i})$ is small, it is a useful predictor in our framework.
\subsection{Other notions of predictor error}
Other geometric notions of error are possible, for example, ones that measure the distance between an agent's true type and the predictor's weakest type. We next show that any $\ell_p$-distance based error is subsumed by our difference-of-welfare error measure.
\begin{proposition}\label{prop:other_errors} Let $\ttheta_i = \argmin_{\theta_i'\in T_i(\btheta_{-i})}w(\theta_i',\btheta_{-i})$. Then, $$|w(\theta_i,\btheta_{-i}) - w(\ttheta_i,\btheta_{-i})|\le \norm{\theta_i - \ttheta_i}_p$$ for all $p\in [1,\infty]$.
\end{proposition}
\begin{proof} Let $\alpha$ and $\widetilde{\alpha}$ denote the efficient allocations on $(\theta_i,\btheta_{-i})$ and $(\ttheta_i,\btheta_{-i})$, respectively.
    We have $w(\theta_i,\btheta_{-i}) - w(\ttheta_i,\btheta_{-i}) = (\theta_i[\alpha] + \sum_{j\neq i}\theta_j[\alpha])-(\ttheta_i[\widetilde{\alpha}] + \sum_{j\neq i}\theta_j[\widetilde{\alpha}]) \le (\theta_i[\widetilde{\alpha}] + \sum_{j\neq i}\theta_j[\widetilde{\alpha}])-(\ttheta_i[\widetilde{\alpha}] + \sum_{j\neq i}\theta_j[\widetilde{\alpha}]) = \theta_i[\widetilde{\alpha}] - \ttheta_i[\widetilde{\alpha}]\le \norm{\theta_i-\ttheta_i}_{\infty}.$ Similarly, $ w(\ttheta_i,\btheta_{-i})-w(\theta_i,\btheta_{-i})\le \ttheta_i[\alpha] - \theta_i[\alpha]\le \norm{\theta_i-\ttheta_i}_{\infty}$. The $\ell_{\infty}$ norm is upper bounded by the $\ell_{p}$ norm for all $p$, so we are done. 
\end{proof}
Proposition~\ref{prop:other_errors} implies that any predictor that has low error under a $\ell_p$-based metric also has low error under our information-rent based metric. This provides further credence to our use of weakest type welfares/information rent as a unifying way to measure the quality of a prediction.

\subsection{Computational considerations}\label{sec:computation}
Before we present our main mechanism and its guarantees, we briefly discuss the computational complexity of computing weakest type payments. We consider the special case where the sets $T_i(\btheta_{-i})$ output by the predictors are polytopes. Let $\textsf{size}(T_i)$ denote the max encoding size over all $\btheta_{-i}$ required to write down the constraints defining $T_i(\btheta_{-i})$.

\begin{theorem}\label{theorem:lp} Let $T_i(\btheta_{-i})$ be a polytope for all $\btheta_{-i}$. Weakest type and corresponding welfare $\min_{\ttheta_i\in T_i(\btheta_{-i})}w(\ttheta_i,\btheta_{-i})$, and thus payment $p_i = \min_{\ttheta_i\in T_i(\btheta_{-i})}w(\ttheta_i,\btheta_{-i})-\sum_{j\neq i}\theta_j[\alpha^*]$, can be computed in $poly(|\Gamma|, \textsf{size}(T_i), n)$ time.
\end{theorem}

\begin{proof}
    Weakest type computation is a min-max optimization problem: $$\min_{\ttheta_i\in T_i(\btheta_{-i})}w(\ttheta_i,\btheta_{-i}) = \min_{\ttheta_i\in T_i(\btheta_{-i})}\max_{\alpha\in\Gamma}\left(\ttheta_i[\alpha] + \sum_{j\neq i}\theta_j[\alpha]\right).$$ We can rewrite the min-max problem as a pure minimization problem by enumerating the set of allocations $\Gamma$ and introducing an auxiliary scalar variable $\gamma$ to replace the inner maximization.
    The weakest type in $T_i(\btheta_{-i})$ is therefore the solution $\ttheta_i\in\R^{\Gamma}$ to the linear program
\begin{equation}\label{eq:LP}\min
\left\{
  \gamma \;:\;
  \begin{aligned}
  & \ttheta_i[\alpha] + \textstyle\sum_{j\neq i}\theta_j[\alpha]\le\gamma \;\;\forall\alpha\in\Gamma,\\
  & \ttheta_i\in T_i(\btheta_{-i}), \gamma\ge 0
  \end{aligned}
\right\}
\end{equation} with $|\Gamma|+1$ variables and $|\Gamma| + \textsf{size}(T_i)$ constraints. Generating the first set of constraints requires the value of $\sum_{j\neq i}\theta_j[\alpha]$ for each $\alpha\in\Gamma$, which takes time at most $n|\Gamma|$ to compute. 
\end{proof}

More generally, the complexity of the above mathematical program is determined by the complexity of constraints needed to define $T_i(\btheta_{-i})$: for example, if $T_i(\btheta_{-i})$ is a convex set then they are convex programs. Naturally, a major caveat of Theorem~\ref{theorem:lp} is that $|\Gamma|$ can be very large (for example, $|\Gamma|$ is exponential in combinatorial auctions). However, this issue can be circumvented as long as we have access to a practically-efficient routine for finding welfare-maximizing allocations, that is, for computing $w(\btheta)$. For small allocation spaces that might amount to an exhaustive search over $\Gamma$. In large allocation spaces like in combinatorial auctions, that might involve integer programming or practically-efficient custom search techniques~\citep{Sandholm05:CABOB, Rothkopf98:Computationally}.

\begin{theorem}\label{theorem:separation_oracle}
    Linear program~\eqref{eq:LP} can be solved with polynomially many calls to $w(\cdot)$ and additional $poly(\textsf{size}(T_i))$ time.
\end{theorem}

\begin{proof}
    Let $\widetilde{\Gamma}\subseteq\Gamma$ denote the set of allocations that appear in the description of $T_i(\btheta_{-i})$ ({\em e.g.}, as a list of linear constraints on $\theta_i$). Then, for any $\alpha\notin\widetilde{\Gamma}$, the weakest type $\ttheta_i$ that minimizes $w(\ttheta_i,\btheta_{-i})$ over $\ttheta_i\in T_i(\btheta_{-i})$ satisfies $\ttheta_i[\alpha] = 0$ as $\ttheta_i[\alpha]$ is unconstrained. Therefore, the number of variables in the linear program~\eqref{eq:LP} is not $|\Gamma|$ but can be bounded by $|\widetilde{\Gamma}|\le\textsf{size}(T_i)$. 
    A separation oracle for the linear program, given as input to the Ellipsoid algorithm~\citep{Grotschel93:Geometric}, can be implemented as follows: a candidate point $(\hat{\gamma},\hat{\theta}_i)$ is feasible if and only if $w(\hat{\theta}_i,\btheta_{-i})\ge \hat{\gamma}$ and $\hat{\theta}_i\in T_i(\btheta_{-i})$; else the efficient allocation achieving welfare $w(\hat{\theta}_i,\btheta_{-i})$ represents a violated constraint. 
\end{proof}

\citet{prasad2026weakest} have developed practical constraint generation routines for solving linear program~\eqref{eq:LP} for combinatorial auctions, where the welfare maximization separation oracle is solved with integer programming.

\section{Main mechanism and its guarantees}\label{sec:main_mechanism}

We now present our main mechanism and analyze the total welfare and revenue it generates. Let $\Delta_i^{\VCG}(\btheta_{-i}) \coloneq \min_{\ttheta_i\in T_i(\btheta_{-i})}w(\ttheta_i,\btheta_{-i}) - w(0,\btheta_{-i})$ be the minimum predicted welfare lift over the zero type. Our mechanism $\cM_{L,\gamma}$ is parameterized by two tunable parameters, a {\em rate} $L_i > 1$ and a {\em stopping distance} $\gamma_i$, per agent. (One notational technicality is that the parameters $\gamma_i$ are chosen during the execution of the mechanism and not before. So, the notation $\cM_{L,\gamma}$ is a slight abuse of notation to maintain readability.) %The full specification is given in Figure~\ref{fig:mechanism}.

\begin{figure}[t]
\begin{tcolorbox}[colback=white, parbox=false]
\underline{Mechanism $\cM_{L,\gamma}$}\\
Input: predictors $T_i:\Theta_{-i}\to\cP(\Theta_i)$ for each agent $i$.

\begin{itemize}
    \item Agents asked to reveal types $\theta_1,\ldots,\theta_n$.
    \item Let $\alpha^* = \argmax_{\alpha\in\Gamma}\sum_{i=1}^n\theta_i[\alpha]$ be the efficient allocation. 
    \item For each agent $i$: 
    \begin{itemize}
        \item If $\Delta_i^{\VCG}(\btheta_{-i}) = 0$: set price $p_i = p_i^{\VCG}$.
        \item Else: set rate $L_i > 1$, stopping distance $\gamma_i\in \left(0, \Delta_i^{\VCG}(\btheta_{-i})\right]$, and price $$p_i = \min_{\ttheta_i\in T_i(\btheta_{-i})}w(\ttheta_i,\btheta_{-i}) -\frac{\Delta_i^{\VCG}(\btheta_{-i})}{L_i^{k_i}} - \sum_{j\neq i}\theta_j[\alpha^*],$$ where $k_i$ is chosen uniformly at random from the set $$\left\{0,1,\ldots,\left\lceil\log_{L_i}\left(\frac{\Delta_i^{\VCG}(\btheta_{-i})}{\gamma_i}\right)\right\rceil\right\}.$$
    \end{itemize} 
    \item Let $\cI = \left\{i : \theta_i[\alpha^*] - p_i\ge 0\right\}.$ If agent $i\notin\cI$, $i$ is excluded and receives zero utility (zero value and zero payment).\footnote{One practical consideration is that this step might require a more nuanced implementation of an ``outside option'' for agents to be indifferent between participating and being excluded versus not participating at all. (We do not pursue this highly application-specific issue in this work.) In auction and matching settings this step is standard: the agent simply receives no items.} 
    If agent $i\in\cI$, $i$ enjoys allocation $\alpha^*$ and pays $p_i$.
\end{itemize}
\end{tcolorbox}
%\caption{Mechanism $\cM_{L,\gamma}$.}
%\label{fig:mechanism}
\end{figure}

$\cM_{L,\gamma}$ receives prediction $T_i(\btheta_{-i})\subseteq\Theta_i$ for each agent $i$. If a prediction is too weak to definitively say that agent $i$ adds any welfare to the system (that is, $\Delta_i^{\VCG}=0$), it is not used. For example, if predictors output the entire ambient type space ($T_i(\btheta_{-i}) = \R_{\ge 0}^{\Gamma}$), or if $T_i(\btheta_{-i})$ is of the form described by Proposition~\ref{prop:qualitative}, $\cM_{L,\gamma}$ implements VCG payments. Otherwise, the mechanism performs an exponential search for the welfare crated by agent $i$'s true type, starting at $w(0, \btheta_{-i})$ ($k_i=0$) and stopping at a distance of $\le\gamma_i$ from the minimum predicted welfare $\min_{\ttheta_i\in T_i(\btheta_{-i})}w(\ttheta_i,\btheta_{-i})$. Parameter $L_i$ controls the rate of that search. Critically, the discretization itself of the exponential search depends on the true types of all other agents $\btheta_{-i}$. Finally, if the final welfare level chosen by $\cM_{L,\gamma}$ for agent $i$ (which has no dependence on agent $i$'s revealed type) ends up being too high, the mechanism has a final step that explicitly excludes that agent enforcing that they receive zero utility in order to prevent an IR violation due to charging that agent more than their value.

Before proceeding to the analysis of $\cM_{L,\gamma}$, we record that it is incentive compatible, individually rational, and welfare consistent.

\begin{theorem}
    $\cM_{L,\gamma}$ is IC, IR, and welfare consistent.
\end{theorem}

\begin{proof}
    $\cM_{L,\gamma}$ is IC for the same reason weakest-type VCG is IC (Theorem~\ref{theorem:rev_optimal})---it is a randomization over IC mechanisms. It is IR by definition: all agents with potential IR violations (those not in $\cI$) do not participate and receive zero utility. $\cM_{L,\gamma}$ is welfare consistent since $\theta_i\in T_i(\btheta_{-i})\Rightarrow\min_{\ttheta_i\in T_i(\btheta_{-i})}w(\ttheta_i,\btheta_{-i})\le w(\theta_i,\btheta_{-i})\Rightarrow\text{all prices }\cM_{L,\gamma}\text{ randomizes over are}\le\theta_i[\alpha^*].$
    
\end{proof}

Finally, let $\cM_{L, 0}$ denote---with a slight abuse of notation as it represents the limiting behavior of $\cM_{L,\gamma}$ as $\gamma_i\downarrow 0$---the deterministic mechanism that sets $p_i = \min_{\ttheta_i\in T_i(\btheta_{-i})}w(\ttheta_i,\btheta_{-i}) - \sum_{j\neq i}\theta_j[\alpha^*]$. So, $\cM_{L, 0}$ uses the predicted sets output by each $T_i$ without modification. $\cM_{L,0}$ is essentially weakest-type VCG with a final step (that sacrifices welfare) to ensure that no agent's IR constraint is violated. When predictors have zero error, that is, $\err_i = 0$, $\cM_{L,0}$ achieves welfare and revenue equal to $\OPT$. Of course this approach is extremely brittle: if $\err_i > 0$ (that is, predictor $T_i$ overestimates agent $i$'s welfare level) the value and payment extracted from agent $i$ both drop to zero (though observe that if $\err_i < 0$ welfare is unaffected and payment decreases linearly). We discuss the sensitivity and precise parameter dependence of $\cM_{L,\gamma}$ shortly---first we pin down its precise welfare and revenue guarantees.

\subsection{Guarantees}\label{sec:main_guarantees}

We now state, prove, and discuss our main welfare and revenue guarantees on $M_{L,\gamma}$. 
%Define $\log^+_2:\R\to\R_{\ge 0}$ by $\log^+_2(x) = 0$ if $x < 1$ and $\log^{+}_2(x) = \log_2(x)$ if $x \ge 1$. 
We abbreviate $\err_i(\btheta_{-i})$ and $\Delta_i^{\VCG}(\btheta_{-i})$ as $\err_i$ and $\Delta_i^{\VCG}$, respectively, for readability. The quantity $\err_i^+ = \max\{0, \err_i\}$ denotes the positive part of $\err_i$.

\begin{theorem}[Agent-wise welfare guarantee]\label{theorem:agentwise_welfare_bound}
    The expected value enjoyed by agent $i$ under $M_{L,\gamma}$ is equal to $$\frac{1+\min\left\{K_i, k_i^*\right\}}{1 + K_i}\cdot\theta_i[\alpha^*],$$ where $$K_i = \left\lceil\log_{L_i}\left(\frac{\Delta_i^{\VCG}}{\gamma_i}\right)\right\rceil\text{ and } k_i^* = \left\lfloor\log_{L_i}\left(\frac{\Delta_i^{\VCG}}{\err_i^+}\right)\right\rfloor.$$
\end{theorem}
\begin{proof}
    Let $K_i = \left\lceil\log_{L_i}\left(\frac{\Delta_i^{\VCG}}{\gamma_i}\right)\right\rceil$ and let $k_i^*$ be the smallest integer such that $$w(\theta_i,\btheta_{-i})\ge \min_{\ttheta_i\in T_i(\btheta_{-i})}w(\ttheta_i,\btheta_{-i}) - \frac{\Delta_i^{\VCG}}{L_i^{k}},$$ so $k_i^* = \left\lfloor\log_{L_i}\left(\frac{\Delta_i^{\VCG}}{\err_i^+}\right)\right\rfloor.$ (If $w(\theta_i,\btheta_{-i})\ge \min_{\ttheta_i\in T_i(\btheta_{-i})} w(\ttheta_i,\btheta_{-i})$, that is, predictor $T_i$ {\em underestimates} the welfare created by $i$, $\err_i^+ = 0$ and $k_i^* = +\infty$ is unbounded.) Agent $i$ participates in $\cM_{L,\gamma}$ if and only if the randomly selected $k_i$ satisfies $k_i\le k_i^*$. So, agent $i$'s expected value in the mechanism is (for brevity let $\ttheta_i$ denote the weakest type) $$\begin{aligned}\E\left[\theta_i[\alpha^*]\cdot\mathbf{1}\left[i\in\cI\right]\right] &= \E\left[\theta_i[\alpha^*]\cdot\mathbf{1}\left[w(\theta_i,\btheta_{-i})\ge w(\ttheta_i,\btheta_{-i}) - \frac{\Delta_i^{\VCG}}{L_i^{k_i}}\right]\right] \\ &= \E\left[\theta_i[\alpha^*]\cdot\mathbf{1}\left[k_i\le k_i^*\right]\right] \\ &= \theta_i[\alpha^*]\cdot\Pr\left(k_i\le k_i^*\right).\end{aligned}$$ and $$\Pr\left(k_i\le k_i^*\right) = \min\left\{1,\frac{1+k_i^*}{1+K_i}\right\},$$ which yields the desired bound. 
\end{proof}
\begin{corollary}[Welfare guarantee]\label{cor:welfare_bound}
    The expected welfare generated by $\cM_{L,\gamma}$ is equal to
    $$\sum_{i=1}^n\frac{1+\min\{K_i,k_i^*\}}{1+K_i}\cdot\theta_i[\alpha^*]\ge\min_i\left\{\frac{1+\min\{k_i^*, K_i\}}{1+K_i}\right\}\OPT.$$
\end{corollary}

\begin{theorem}[Agent-wise payment guarantee]\label{theorem:agentwise_payment_bound}
    The expected payment made by agent $i$ in $\cM_{L,\gamma}$ is at least
    $$\E[p_i]\ge \frac{1+\min\{K_i, k_i^*\}}{1+K_i}\left(\theta_i[\alpha^*] + \err_i\right) - \frac{\Delta_i^{\VCG}L_i}{(L_i-1)(K_i+1)},$$ where $K_i$ and $k_i^*$ are defined in Theorem~\ref{theorem:agentwise_welfare_bound}.
\end{theorem}
\begin{proof}
    Let $k_i^*$ be defined as in Theorems~\ref{theorem:agentwise_welfare_bound} and let $\ttheta_i$ be the weakest type in $T_i(\btheta_{-i})$. We have
    $$\begin{aligned}
        \E[p_i] & = \sum_{k = 0}^{\min\{K_i,k_i^*\}}\E[p_i | k_i = k]\cdot\Pr(k_i = k) \\ & = \frac{1}{1+K_i}\sum_{k = 0}^{\min\{K_i,k_i^*\}}\left(\theta_i[\alpha^*] - \left(w(\theta_i,\btheta_{-i}) - \left(w(\ttheta_i,\btheta_{-i}) - \frac{\Delta_i^{\VCG}}{L_i^{k}}\right)\right)\right) \\ 
        &= \frac{1}{1+K_i}\sum_{k=0}^{\min\{K_i, k_i^*\}}\left(\theta_i[\alpha^*] + \err_i - \frac{\Delta_i^{\VCG}}{L_i^k}\right) \\ 
        &= \frac{1+\min\{K_i, k_i^*\}}{1+K_i}\left(\theta_i[\alpha^*] + \err_i\right) - \frac{\Delta_i^{\VCG}}{1+K_i}\sum_{k=0}^{\min\{K_i, k_i^*\}} L_i^{-k},
    \end{aligned}$$ where in the first line we have rewritten the price formula as pay-as-bid with discount.
    We have $$\sum_{k=0}^{\min\{K_i, k_i^*\}} L_i^{-k} < \sum_{k=0}^{\infty} L_i^{-k} = \frac{L_i}{L_i-1}.$$ Substituting yields the desired payment bound. 
\end{proof}
\begin{corollary}[Revenue guarantee]\label{cor:revenue_bound}
The expected revenue generated by $\cM_{L,\gamma}$ is at least $$\min_i\left\{\frac{1+\min\{K_i, k_i^*\}}{1+K_i}\right\}\left(\OPT + \sum_{i=1}^n\err_i\right) - \sum_{i=1}^n\frac{\Delta_i^{\VCG}L_i}{(L_i-1)(K_i+1)}.$$
\end{corollary}

In the above bounds, the term $$\frac{1+\min\{K_i, k_i^*\}}{1+K_i} = \min\left\{1,\frac{1+k_i^*}{1+K_i}\right\}=\min\left\{1, \frac{1+\lfloor\log_{L_i}(\Delta_i^{\VCG}/\err_i^+)\rfloor}{1 + \lceil\log_{L_i}(\Delta_i^{\VCG}/\gamma_i)\rceil}\right\}$$ is the probability that the modified weakest-type welfare is less than or equal to the true welfare created by agent $i$, that is, the probability that $i\in\cI$. Ignoring rounding, this probability is equal to $1$ roughly when $\err_i^+\le \gamma_i$. In that case, the payment/revenue bounds suffer from an additive loss that has no dependence on $\err_i$. This can be interpreted a penalty for how quickly the discretization covers the welfare interval $[w(0,\btheta_{-i}), w(\ttheta_i,\btheta_{-i})]$ (which enables more ``smoothness" in payment degradation) controlled by $L_i$ and $\gamma_i$.

We illustrate how welfare and revenue degrade (restricted to a single agent) as the error of the predictors worsen. We plot agent value (Figure~\ref{fig:value}) and payment (Figure~\ref{fig:payment_plots}) as a function of $\err_i$ for different $(L_i,\gamma_i)$ settings. The key takeaway is that if predictor $T_i$ underestimates welfare, that is $\min_{\ttheta_i\in T_i(\btheta_{-i})}w(\ttheta_i,\btheta_{-i})\le w(\theta_i,\btheta_{-i})$ and $\err_i \le 0$, agent value remains optimal (Figure~\ref{fig:value}) and payment degrades {\em linearly} in $\err_i$ (Figure~\ref{fig:payment_plots}). If predictor $T_i$ overestimates welfare, that is, $\min_{\ttheta_i\in T_i(\btheta_{-i})}w(\ttheta_i,\btheta_{-i})> w(\theta_i,\btheta_{-i})$ and $\err_i > 0$, value (Figure~\ref{fig:value}) and payment (Figure~\ref{fig:payment_plots}) degrade at a rate determined by $L_i,\gamma_i$. Holding $L_i$ constant and increasing $\gamma_i$ leads to a more gradual decay. Smaller values of $\gamma_i$ yield greater payment extracted when $\err_i\le 0$, but lead to more drastic payment degradation for $\err_i > 0$. A similar effect is illustrated while holding $\gamma_i$ constant and increasing the search rate $L_i$. Both parameters represent trade-offs between performance in the best case and error tolerance, and is one that the mechanism designer must choose carefully based on his confidence in the prediction. Data-driven algorithm design~\citep{balcan2020data} provides a toolkit for parameter tuning with provable guarantees. While we do not provide specific guidance on parameter selection in this work, an important future research question is to characterize the optimal tuning based on other knowledge about predictor quality (for example, the mechanism designer might have a prior over prediction error).

\begin{figure}[t]
    \centering
    \includegraphics[width=0.75\linewidth, trim={0cm 0cm 0cm, 0cm},clip  ]{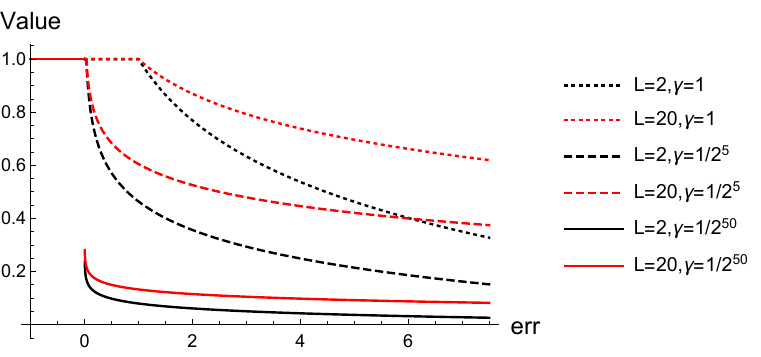}
    \caption{An agent's expected value (as a fraction of $\theta_i[\alpha^*]$) as a function of $\err_i$ for problem parameter $\Delta_i^{\VCG} = 10$, varying $\gamma_i\in\{1, 2^{-5}, 2^{-50}\}$ and $L_i\in\{2, 20\}$.}
    \label{fig:value}
\end{figure}

\begin{figure}[t]
    \centering
    \includegraphics[width=0.8\linewidth, trim={0cm 0cm 0cm, 0cm},clip  ]{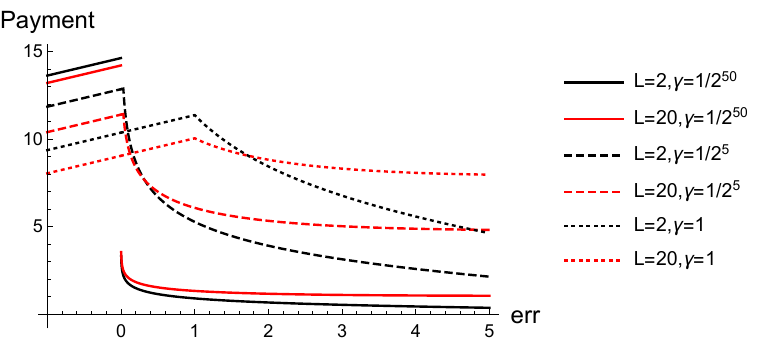}
    \caption{Payment as a function of $\err_i$ for problem parameters $\theta_i[\alpha^*] = 15$, $\Delta_i^{\VCG} = 10$, varying $\gamma_i\in\{1, 2^{-5}, 2^{-50}\}$ and $L_i\in\{2,20\}$.}
    \label{fig:payment_plots}
\end{figure}

\subsection{Consistency, robustness, and barriers to our approach}\label{sec:consistency_and_robustness}
We situate our results within the consistency-robustness framework studied by the algorithms-with-predictions literature. Furthermore, we retrospectively discuss the failures of some alternate approaches---and how $\cM_{L,\gamma}$ addresses those failures---that are solely concerned with consistency and robustness measures. Finally, we establish impossibility results that show that our mechanism's performance cannot be substantially improved.

We say a mechanism is $(a, b)$-{\em consistent} and $(c, d)$-{\em robust} if when predictions are perfect (which, in our setting, means $\Delta_i^{\err} = 0 \iff \min_{\ttheta_i\in T_i(\btheta_{-i})}w(\ttheta_i,\btheta_{-i}) = w(\theta_i,\btheta_{-i})$) it satisfies $\E[\text{welfare}]\ge a\cdot\OPT$, $\E[\text{revenue}]\ge b\cdot\OPT$, and satisfies $\E[\text{welfare}]\ge c\cdot\OPT,\E[\text{revenue}]\ge d\cdot\textsf{VCG}$ independent of the prediction quality (where $\textsf{VCG}$ denotes the revenue of the vanilla VCG mechanism). Consistency demands near-optimal performance when the side information is perfect, and therefore we compete with the total social surplus $\OPT$ on both the welfare and revenue fronts. Robustness deals with the case of arbitrarily bad side information, in which case we would like our mechanism's performance to be competitive with vanilla VCG, which already obtains welfare equal to $\OPT$. High consistency and robustness ratios are in fact trivial to achieve, and we will thus largely not be too concerned with these measures---our main goal is to design high-performance mechanisms that degrade gracefully as the prediction errors increase. 

\paragraph{Failures of other approaches}
We discuss a trivial approach that obtains high consistency and robustness ratios, but suffers from huge discontinuous drops in performance even when predictions are nearly perfect, further illustrating the need for a well-tuned instantiation of $\cM_{L,\gamma}$.

    {\em Trust predictions completely.} One trivial way of using predictions is to trust them completely, that is, run weakest-type VCG with payments given by $p_i = \min_{\ttheta_i\in T_i(\btheta_{-i})}w(\ttheta_i,\btheta_{-i})-\sum_{j\neq i}\theta_j[\alpha^*]$ and exclude any agent who is charged too much (this mechanism is given by $\cM_{L,0}$). This approach generates welfare $\sum_{i : \err_i\le 0}\theta_i[\alpha^*]$ and revenue $\sum_{i:\err_i\le  0}\theta_i[\alpha^*] - |\err_i|$ (Theorem~\ref{theorem:revenue_characterization}). If predictions are perfect, that is, $\err_i = 0$ for all $i$, both welfare and revenue are equal to $\OPT$. However, if all predictions are such that $\err_i > 0$, both welfare and revenue drop to $0$. So this mechanism is $(1,1)$-consistent and $(0,0)$-robust.
    
    {\em Discard predictions randomly.} The issue with the above mechanism is that if all predictions are invalid, it generates no welfare and no revenue. We show how randomization can quell that issue. One trivial solution is to discard all predictions with probability $\beta$, and trust all predictions completely with probability $(1-\beta)$. That is, with probability $\beta$ charge each agent her vanilla VCG price and with probability $1-\beta$ charge each agent her weakest-type price (and exclude any agent who is overcharged). This mechanism achieves strong consistency and robustness ratios. Indeed, its expected welfare is $\beta\cdot\OPT + (1-\beta)\cdot\sum_{i:\err_i\le 0}\theta_i[\alpha^*]$ and its expected revenue is $\beta\cdot\VCG + (1-\beta)\cdot (\sum_{i:\err_i\le 0}\theta_i[\alpha^*] - |\err_i|)$. So, it is $(1, 1-\beta)$-consistent and $(\beta,\beta)$-robust.
    This approach suffers from a major issue: its revenue drops drastically the moment predictions overestimate welfare, that is, $\err_i > 0$. In particular, if predictions overestimate an agent's welfare level, {\em but barely so}, this approach completely misses out on any payments from such agents and drops to the revenue of VCG (which can be drastically smaller than $\OPT$). But, a tiny relaxation of these predictions would have sufficed to increase revenue significantly and perform competitively with $\OPT$. One simple approach is to subtract a relaxation parameter $\eta_i$ from the welfare of the weakest type consistent with each prediction with some probability, and discard the prediction with complementary probability. If $\err_i - \eta_i \le 0$ for all $i$, then such a mechanism would perform well. The main issue with such an approach is that the moment $\err_i - \eta_i > 0$, our relaxation by $\eta_i$ is still a welfare overestimate and the performance drastically drops. Our main mechanism $\cM_{L,\gamma}$ essentially selects the $\eta_i$ randomly from a suitable discretization of the possible welfare levels to ``smooth out" this behavior and extract payments from each agent with reasonable probability.

\paragraph{Consistency and robustness of $\cM_{L,\gamma}$}

We determine the consistency and robustness properties of our mechanism. While the usual notion of consistency asks for a multiplicative comparison against $\OPT$, our payment guarantee will differ additively from $\OPT$. This is a stronger style of consistency result.

\begin{theorem}[Consistency]\label{theorem:consistency}
    Suppose $\min_{\ttheta_i\in T_i(\btheta_{-i})} w(\ttheta_i,\btheta_{-i}) = w(\theta_i,\btheta_{-i})$. Then, $\cM_{L,\gamma}$ is efficient (that is, it's welfare is equal to $\OPT$), and agent $i$'s expected payment is $$\E[p_i]\ge\theta_i[\alpha^*] - \frac{L_i(w(\theta_i,\btheta_{-i}) - w(0, \btheta_{-i}))}{(L_i-1)\left(1 + \left\lceil\log_{L_i}\left(\frac{w(\theta_i,\btheta_{-i}) - w(0, \btheta_{-i})}{\gamma_i}\right)\right\rceil\right)}.$$
\end{theorem}
\begin{proof}
    Immediate from Theorems~\ref{theorem:agentwise_welfare_bound} and~\ref{theorem:agentwise_payment_bound}. 
\end{proof}
So, when predictors are perfect, $\cM_{L,\gamma}$ is effectively able to reduce the discount given by vanilla VCG by a term $\log_{L_i}((w(\theta_i,\btheta_{-i}) - w(0,\btheta_{-i}))/\gamma_i)$ that increases as the stopping distance $\gamma_i$ decreases. For $(L_i,\gamma_i)$ pairs such that $L_i\ge 2$ and $$\gamma_i\le\frac{w(\theta_i,\btheta_{-i}) - w(0,\btheta_{-i})}{L_i^{2(w(\theta_i,\btheta_{-i}) - w(0,\btheta_{-i}))/\varepsilon}},$$ the bound reads $\E[p_i]\ge\theta_i[\alpha^*] - \varepsilon$. As $\gamma_i$ tends to zero, the limiting behavior $\cM_{L,\gamma}$ is that of the mechanism that trusts predictors completely, that is, it deterministically uses prices $p_i = \min_{\ttheta_i\in T_i(\btheta_{-i})} w(\ttheta_i,\btheta_{-i})$, which is the only revenue consistent mechanism.

\begin{theorem}[Robustness of a predictor that underestimates welfare]\label{theorem:robustness_underestimate} Suppose $\min_{\ttheta_i\in T_i(\btheta_{-i})} w(\ttheta_i,\btheta_{-i})$ $\le w(\theta_i,\btheta_{-i})$. Then, $\cM_{L,\gamma}$ generates welfare equal to $\OPT$, and agent $i$'s expected payment is
    $$\E[p_i]\ge \max\left\{p_i^{\VCG},\theta_i[\alpha^*] - |\err_i| - \frac{L_i\Delta_i^{\VCG}}{(L_i-1)\left(1+ \left\lceil\log_{L_i}\left(\Delta_i^{\VCG}/\gamma_i\right)\right\rceil\right)}\right\}.$$
\end{theorem}
\begin{proof}
    If $\min_{\ttheta_i\in T_i(\btheta_{-i})}w(\ttheta_i,\btheta_{-i})\le w(\theta_i,\btheta_{-i})$, every modified welfare that $\cM_{L,\gamma}$ randomizes over is also less than $w(\theta_i,\btheta_{-i})$. So, $\cM_{L,\gamma}$ generates the optimal welfare independent of the value of $\err_i \le 0$. (So $\cM_{L,\gamma}$ actually satisfies a stronger condition than welfare consistency.) The payment bound is immediate from Theorem~\ref{theorem:agentwise_payment_bound}. 
\end{proof}
\begin{theorem}[Robustness of a predictor that overestimates welfare]\label{theorem:robustness_overestimate} Suppose $\min_{\ttheta_i\in T_i(\btheta_{-i})}w(\ttheta_i,\btheta_{-i})$ $> w(\theta_i,\btheta_{-i})$. Then, the expected value generated from agent $i$ is at least $$\frac{1}{1+\left\lceil\log_{L_i}\left(\Delta_i^{\VCG}/\gamma_i\right)\right\rceil}\theta_i[\alpha^*],$$ and agent $i$'s expected payment is 
    $$\E[p_i]\ge\frac{1}{1+\left\lceil\log_{L_i}\left(\Delta_i^{\VCG}/\gamma_i\right)\right\rceil}\max\left\{p_i^{\VCG}, \theta_i[\alpha^*] - \err_i^+(L_i - 1)\right\}.$$
\end{theorem}
\begin{proof}
    We have $$\E[p_i]\ge \E[p_i\mid k_i = k_i^*]\Pr(k_i=k_i^*) = \frac{1}{1+K_i}\left(\theta_i[\alpha^*] - \left( w(\theta_i,\btheta_{-i}) - \left(w(\ttheta_i,\btheta_{-i}) - \frac{\Delta_i^{\VCG}}{L_i^{k_i^*}}\right)\right)\right).$$ Since $k_i^*$ is minimal such that $w(\theta_i,\btheta_{-i})\ge w(\ttheta_i,\btheta_{-i}) - \frac{\Delta_i^{\VCG}}{L_i^{k_i^*}}$, we have $$\begin{aligned}w(\theta_i,\btheta_{-i}) - \left(w(\ttheta_i,\btheta_{-i}) - \frac{\Delta_i^{\VCG}}{L_i^{k_i^*}}\right) &\le \left(w(\ttheta_i,\btheta_{-i}) - \frac{\Delta_i^{\VCG}}{L_i^{k_i^*+1}}\right) - \left(w(\ttheta_i,\btheta_{-i}) - \frac{\Delta_i^{\VCG}}{L_i^{k_i^*}}\right) \\ &= \Delta_i^{\VCG}\left(L_i^{-k_i^*} - L_i^{-(k_i^*+1)}\right) \\ &= \Delta_i^{\VCG}(L_i - 1)L_i^{-(k_i^*+1)} \\ &= \Delta_i^{\VCG}(L_i-1)\left(L_i^{1+\lfloor\log_{L_i}(\Delta_i^{\VCG}/\err_i^+)\rfloor}\right)^{-1} \\ &\le \Delta_i^{\VCG}(L_i-1)\left(L_i\cdot L_i^{\log_{L_i}(\Delta_i^{\VCG}/\err_i^+)-1}\right)^{-1} \\ &= \Delta_i^{\VCG}(L_i-1)\left(\Delta_i^{\VCG}/\err_i^+\right)^{-1} \\ &= (L_i-1)\err_i^+. \end{aligned}$$ So, $\E[p_i]\ge \frac{1}{1+K_i}\left(\theta_i[\alpha^*] - \err_i^+(L_i-1)\right)$. Finally, there is always a $\frac{1}{1+K_i}$ probability that $k_i = 0$ in which case $p_i = p_i^{\VCG}$. 
\end{proof}

\begin{remark} A mechanism that achieves perfect revenue consistency, that is, $\E[p_i] = \theta_i[\alpha^*]$ whenever $T_i(\btheta_{-i}) = \theta_i$ cannot be error tolerant. Indeed, the only way of attaining perfect revenue consistency is to trust the minimum predicted welfare with no randomness or modification. That mechanism obtains zero payment if a predictor overestimates welfare by any amount.\end{remark}

Consistency and robustness do not capture the important effect of how aggressive or how conservative a prediction is on the performance decay of our mechanism, as explained in Section~\ref{sec:main_guarantees} and the plots in Figures~\ref{fig:value} and~\ref{fig:payment_plots}. The mechanism's performance can furthermore be greatly improved through high quality hyperparameter selection, which can, for example, be learned from data~\citep{khodak2022learning,balcan2020data}.

\paragraph{Impossibility results} There are natural barriers to how much revenue can be obtained in multidimensional mechanism design settings. For multi-unit auctions,~\citet{sandholm2015automated} construct a distribution over additive valuations such that for any mechanism, its expected revenue is at most $(1 - \frac{1}{h/\ell})\frac{1}{\ln(h/\ell)}$ times its expected social welfare, where $h$ and $\ell$ are upper and lower bounds on any bidders' valuation for any item (and are known to the mechanism designer). Furthermore, they show that for any mechanism that is completely prior free (so it does not even have knowledge of bounds on valuations like $h$ and $\ell$) and any $\varepsilon>0$, there exist distributions over valuations such that the mechanism's expected revenue is at most $\varepsilon$ times its expected welfare. This result is obtained by studying a single-item setting. Both of these barriers carry over to our setting as well, since the setting studied by~\citet{sandholm2015automated} is a special case of ours.
We now present a couple results along a similar vein in our prediction-augmented setting. First, we show how revenue optimality of weakest type (Theorem~\ref{theorem:rev_optimal}) implies an upper bound on the revenue of any welfare consistent mechanism. Then, we present a construction akin to that of~\citet{sandholm2015automated}. Like~\citet{sandholm2015automated}, our construction relies on a ``hard" distribution over single-parameter agents. It replaces the dependence on bounds $h$ and $\ell$---which we do not assume the mechanism designer has knowledge of---with qualities of the predictors input to the mechanism. Below, we call a predictor $T_i : \bTheta_{-i}\to\cP(\Theta_i)$ {\em constant} if $T_i$ is a constant function (and has no dependence on its input $\btheta_{-i}$).
\begin{theorem}
    Any welfare consistent mechanism satisfies $\E[p_i(\btheta)]\le \min_{\ttheta_i\in T_i(\btheta_{-i})} w(\ttheta_i,\btheta_{-i}) - \sum_{j\neq i}\theta_j[\alpha^*]$, where the expectation is over any randomness in the mechanism.
\end{theorem}
\begin{proof}
    Welfare consistency means that for any $\btheta_{-i}$, IC constraints, IR constraints, and efficiency must hold over $T_i(\btheta_{-i})$. By Theorem~\ref{theorem:rev_optimal}, weakest-type payments $\min_{\ttheta_i\in T_i(\btheta_{-i})} w(\ttheta_i,\btheta_{-i}) - \sum_{j\neq i}\theta_j[\alpha^*]$ maximize agent $i$'s payment. 
\end{proof}
\begin{theorem}\label{theorem:main_impossibility}
    For any constant predictors $T_1,\ldots, T_n$ such that $\min_i\min_{\ttheta_i\in T_i}\max_{\alpha\in\Gamma}\ttheta_i[\alpha] > 1$, there exists a distribution $D$ over agent types $\btheta$ such that for any welfare consistent (possibly randomized) Groves mechanism, $$\frac{\E_{\btheta\sim D}[p_i(\btheta)]}{\E_{\btheta\sim D}[\theta_i[\alpha^*(\btheta)]]}\le\frac{1}{1 + \ln\left(\min_{\ttheta_i\in T_i}\max_{\alpha\in\Gamma}\ttheta_i[\alpha]\right)}$$ for all agents $i$, and $$\frac{\E_{\btheta\sim D}[\text{revenue}]}{\E_{\btheta\sim D}[\text{welfare}]}\le\frac{n}{n + \sum_{i=1}^n \ln\left(\min_{\ttheta_i\in T_i}\max_{\alpha\in\Gamma}\ttheta_i[\alpha]\right)} .$$
\end{theorem}
\begin{proof}
    Fix an allocation $\alpha\in\Gamma$. Agent $i$'s type is such that $\theta_i[\alpha'] = 0$ for all $\alpha'\neq\alpha$, and $\theta_i[\alpha]$ is drawn from a truncated equal-revenue distribution on the interval $\big[1, \min_{\ttheta_i\in T_i}\ttheta_i[\alpha]\big].$ Formally, $\theta_i[\alpha]$ has the cdf $F(x) = \Pr(\theta_i[\alpha]\le x)$ given by $$F(x) = \begin{cases} 0 & x < 1 \\ 1 - 1/x & x\in \left[1, \min_{\ttheta_i\in T_i}\ttheta_i[\alpha]\right) \\ 1 & x\ge \min_{\ttheta_i\in T_i}\ttheta_i[\alpha]\end{cases}.$$ Types $\theta_1,\ldots, \theta_n$ are drawn independently. So, the efficient allocation is always $\alpha$, independent of the realization of agent types. Agent $i$'s expected efficient value is $$\E[\theta_i[\alpha]] = \int_{0}^{ \min_{\ttheta_i\in T_i}\ttheta_i[\alpha]} (1 - F(x))dx = 1 + \int_1^{\min_{\ttheta_i\in T_i}\ttheta_i[\alpha]} (1/x) dx = 1 + \ln\left(\min_{\ttheta_i\in T_i}\ttheta_i[\alpha]\right).$$ The expected efficient welfare is therefore $\E[\sum_i\theta_i[\alpha]] = \sum_i E[\theta_i[\alpha]] = n + \sum_{i=1}^n\ln(\min_{\ttheta_i\in T_i}\ttheta_i[\alpha])$.
    Now, for any deterministic Groves mechanism $\{h_i(\btheta_{-i})\}$ with $h_i(\btheta_{-i})\le\min_{\ttheta_i\in T_i}\ttheta_i[\alpha]+\sum_{j\neq i}\theta_j[\alpha]$ (which is required by welfare consistency), its expected payment captured from agent $i$ is $$\begin{aligned}\E_{\btheta}[p_i] &= \E_{\btheta}\left[\left(h_i(\btheta_{-i}) - \sum_{j\neq i}\theta_j[\alpha]\right)\mathbf{1}(w(\theta_i,\btheta_{-i})\ge h_i(\btheta_{-i}))\right] \\ &= \E_{\btheta_{-i}}\left[\E_{\theta_i}\left[\left(h_i(\btheta_{-i}) - \sum_{j\neq i}\theta_j[\alpha]\right)\mathbf{1}\left(\theta_i[\alpha]\ge h_i(\btheta_{-i}) - \sum_{j\neq i}\theta_j[\alpha]\right)\;\,\Bigg\lvert\;\,\btheta_{-i}\right]\right] \\ &= \E_{\btheta_{-i}}\left[\left(h_i(\btheta_{-i}) - \sum_{j\neq i}\theta_j[\alpha]\right)\Pr\left(\theta_i[\alpha]\ge h_i(\btheta_{-i}) - \sum_{j\neq i}\theta_j[\alpha]\right)\right] \\ &= 1,\end{aligned}$$ where the second equality follows from the law of total expectation and in the final equality we have used the fact that, given $\btheta_{-i}$, $\Pr_{\theta_i}(\theta_i[\alpha]\ge h_i(\btheta_{-i}) - \sum_{j\neq i}\theta_j[\alpha]) = \frac{1}{h_i(\btheta_{-i}) - \sum_{j\neq i}\theta_j[\alpha]}$. It follows that the expected revenue of any perfectly consistent Groves mechanism is $n$. The welfare and revenue computations extend to randomized Groves mechanisms since they hold point-wise over deterministic ones. 
\end{proof}
Theorems~\ref{theorem:robustness_overestimate} and~\ref{theorem:main_impossibility} together demonstrate that the performance of $\cM_{L,\gamma}$ cannot be significantly improved.

\section{Other forms of side information}\label{sec:other_forms}

We apply the weakest-type VCG mechanism to three other formats of side information that are distinct from the model of predictors used in the paper thus far. In each format, the weakest types are instantiated in a different way.

\subsection{A more expressive prediction language for expressing uncertainty}\label{sec:other_forms+uncertainty}

In this subsection we establish an avenue for richer and more expressive side information languages. We show that the techniques we have developed so far readily extend to an even larger more expressive form of side information that allows one to express varying degrees of uncertainty. We now allow the output of $T_i(\btheta_{-i})$ to be an entire probability space $(\Theta_i, \cF_i, \mu_i)$ where agent $i$'s ambient type space $\Theta_i$ is the sample space, $\cF_i$ is a $\sigma$-algebra on $\Theta_i$, and $\mu_i$ is a probability measure. As an example, the scenarios given in the introduction can be modified to express quantiles of uncertainty. A side information structure for quantiles would have the form: ``Agent $i$'s type satisfies property $P_1$ with probability $p_1$, property $P_2\subset P_1$ with probability $p_2 < p_1$, and so on."

$\cF_i$ induces a partition of $\Theta_i$ into equivalence classes where $\theta_i\equiv\theta'_i$ if $\mathbf{1}[\theta_i\in A] = \mathbf{1}[\theta'_i\in A]$ for all $A\in\cF_i$ (so the side-information structure does not distinguish between $\theta_i$ and $\theta'_i$). Let $A(\theta_i)=\{\theta'_i : \theta_i\equiv\theta'_i\}\in\cF_i$ be the equivalence class of $\theta_i$. In this way the $\sigma$-algebra $\cF_i$ determines the granularity of knowledge being conveyed by the predictor, and the probability measure $\mu_i:\cF_i\to[0,1]$ establishes uncertainty over this knowledge. Our model of side information in the form of a prediction set $T_i(\btheta_{-i}) = T_i$ considered previously in the paper corresponds to the $\sigma$-algebra $\cF_i = \{\emptyset, T_i, \Theta_i\setminus T_i, \Theta_i\}$ with $\mu_i(\emptyset) = \mu_i(\Theta_i\setminus T_i) = 0$ and $\mu_i(T_i) = \mu_i(\Theta_i) = 1$.

We define the error of a predictor in the natural way. As usual, $\btheta$ denotes the agents' (true and) revealed type profile. Define random variable $X_i^{\err}:\Theta_i\to\R_{\ge 0}$ by $$X_i^{\err}(\theta_i') = w(\theta_i,\btheta_{-i}) - \min_{\ttheta_i\in A(\theta_i')}w(\ttheta_i,\btheta_{-i}).$$ $X_i^{\err}$ is $\cF_i$-measurable since it is (by definition) constant on all atoms of $\cF_i$ (sets $A\in\cF_i$ such that no nonempty $B\subset A$ is in $\cF_i$). The error distribution on $\R_{\ge 0}$ is given by $$\begin{aligned}\Pr(a\le X_i^{\err}\le b) &= \mu_i\left(\left\{\theta_i'\in\Theta_i : a\le w(\theta_i,\btheta_{-i}) - \min_{\ttheta_i\in A(\theta_i')}w(\ttheta_i,\btheta_{-i})\le b\right\}\right) \\ &= \mu_i\left(\bigcup\left\{A(\theta_i'):a\le w(\theta_i,\btheta_{-i}) - \min_{\ttheta_i\in A(\theta_i')}w(\ttheta_i,\btheta_{-i})\le b\right\}\right).\end{aligned}$$

(Prior work on mechanism design with predictions has focused on deterministic predictions with deterministic errors. Allowing for a full probability space, therefore inducing error random variables, is a significant departure from, and increase in generality over, the classic algorithms-with-predictions paradigm.)

The generalized version of $M_{L,\gamma}$ that receives as input a generalized predictor for each agent $i$ given by $T_i(\btheta_{-i}) = (\Theta_i, \cF_i, \mu_i)$ works as follows. It samples $\theta_i'\sim\Theta_i$ according to $(\cF_i,\mu_i)$, sets $\ttheta_i = \argmin_{\hat{\theta}_i\in A(\theta_i')}w(\hat{\theta}_i,\btheta_{-i})$, and draws $k_i\sim_{\text{unif.}}\{0,\ldots, \lceil\log_L\Delta_i^{\VCG}/\gamma_i\rceil\}$, where $\Delta_i^{\VCG} = w(\ttheta_i,\btheta_{-i}) - w(0,\btheta_{-i})$, as before. It then implements the efficient allocation $\alpha^*$ and computes a payment for agent $i$ of $p_i = w(\ttheta_i,\btheta_{-i}) - \Delta_i^{\VCG}/L_i^{k_i} - \sum_{j\neq i}\theta_j[\alpha^*]$, excluding agents for which $p_i > \theta_i[\alpha^*]$. Applying Theorems~\ref{theorem:agentwise_welfare_bound} and~\ref{theorem:agentwise_payment_bound} for a fixed $\theta_i'$ and then taking expectation over the draw of $\theta_i'$ yields the following guarantees. 

\begin{theorem}\label{theorem:uncertainty_guarantees}
    The expected value enjoyed by agent $i$ is equal to
    $$\E_{\ttheta_i}\left[\min\left\{1,\frac{1+\lfloor\log_{L_i}(\Delta_i^{\VCG}/\err^+_i)\rfloor}{1+\lceil\log_{L_i}(\Delta_i^{\VCG}/\gamma_i)\rceil}\right\}\right]\theta_i[\alpha^*]$$ and the expected payment made by agent $i$ is at least $$\E_{\ttheta_i}\left[\min\left\{1,\frac{1+\lfloor\log_{L_i}(\Delta_i^{\VCG}/\err^+_i)\rfloor}{1+\lceil\log_{L_i}(\Delta_i^{\VCG}/\gamma_i)\rceil}\right\}(\theta_i[\alpha^*] +\err_i)-\frac{L_i \Delta_i^{\VCG} }{(L_i-1)(1+\lceil\log_{L_i}(\Delta_i^{\VCG}/\gamma_i)\rceil)}\right].$$
\end{theorem}

Despite the richness of information conveyed by an entire probability space, the above mechanism, for each agent, draws only a single sample according to $(\cF_i,\mu_i)$. An interesting research question is whether or not more samples can be used to improve its performance, or if there is a way for the mechanism to truly use all the information conveyed by $(\cF_i,\mu_i)$.

Finally, the above discussion assumes the existence of a routine for sampling from the abstract probability space specified by the predictors. We briefly discuss a concrete special form of generalized predictors for which this routine can be concretely described. For each agent $i$, $T_i(\btheta_{-i})$ outputs (i) a partition $(A^i_1,\ldots, A^i_m)$ of the ambient type space $\Theta_i$ into disjoint sets, (ii) probabilities $\mu^i_1,\ldots, \mu^i_m\ge 0; \sum_j\mu^i_j = 1$ corresponding to each partition element, and (iii) for each partition element an optional probability density function $f^i_j; \int_{A_j^i}f^i_j = 1$. The prediction represents (i) a belief over what partition element $A_j^i$ the true type $\theta_i$ lies in and (ii) if a density is specified, the precise nature of uncertainty over the true type within $A_j^i$. Our model of side information in the form of a prediction set $T_i(\btheta_{-i}) = T_i \subseteq\Theta_i$ considered earlier in the paper corresponds to the partition $(T_i, \Theta_i\setminus\ T_i)$ with $\mu(T_i) = 1$ and no specified densities. The richer model allows side information to convey finer-grained beliefs; for example one can express quantiles of certainty, precise distributional beliefs, and arbitrary mixtures of these. Here, $M_{L,\gamma}$ first samples a partition element $A_j^i$ according to $(\mu_1^i,\ldots,\mu_m^i)$, and draws $k_i\sim_{\text{unif.}}\{0,\ldots, K_i\}$ where $K_i$ is defined as before. If $f_j^i=\texttt{None}$, it uses weakest type $\ttheta_i$ that minimizes $w(\ttheta_i,\btheta_{-i})$ over $\ttheta_i\in A_j^i$. Otherwise, it samples $\ttheta_i\sim f_j^i$ and uses that as the weakest type.

\subsection{Constant-parameter agents: types on low-dimensional subspaces}
\label{section:other_forms+subspace}

In this section we show how the theory we have developed so far can be used to derive new revenue approximation results when the mechanism designer knows that each agent's type belongs to some low-dimensional subspace of her ambient type space $\R^{\Gamma}_{\ge 0}$ (these subspaces can be different for each agent). 

This is a different setup from the previous sections. So far, we have assumed that $\Theta_i = \R^{\Gamma}_{\ge 0}$ for all $i$, that is, there is an ambient type space that is common to all the agents. Side information in the form of predictors $T_i:\bigtimes_{j\neq i}\R_{\ge 0}^{\Gamma}\to\R^{\Gamma}_{\ge 0}$ are given as input to the mechanism designer, with no assumptions on quality/correctness (and our guarantees in Section~\ref{sec:main_guarantees} were parameterized by the quality of the predictors). Here, we assume the side information---which says that each agent's type lies in a particular subspace---is guaranteed to be valid. Two equivalent ways of stating this setup are (1) that $\Theta_i$ is the corresponding subspace for agent $i$ and the mechanism designer receives no additional predictor $T_i$ or (2) $\Theta_i =\R^{\Gamma}_{\ge 0}$ and $T_i(\btheta_{-i}) = \R^{\Gamma}_{\ge 0}\cap U_i$ where $U_i$ is a subspace of $\R^{\Gamma}$, and the mechanism designer has the additional guarantee that $\theta_i\in U_i$ (so $T_i$ is a \emph{correct} predictor in that the set it outputs actually contains the agent's true type). We shall use the language of the second interpretation.

In this setting, while the predictors are correct (which implies $\err_i(\btheta_{-i})\le 0$), their errors $|\err_i|$ can be too large in magnitude to meaningfully use our previous guarantees. More precisely, since $T_i$ outputs an entire linear subspace of the ambient type space, it contains low welfare types---in particular it contains the zero type that creates welfare $w(0, \btheta_{-i})$. So, our mechanism from Section~\ref{sec:main_mechanism} is not useful here as its revenue will be no better than vanilla VCG.

In this section we show how to fruitfully use the information provided by the subspaces $U_1,\ldots,U_n$ within the framework of our meta-mechanism. We assume $\Theta_i=[1,H]^{\Gamma}$, thereby imposing a lower bound of $1$ (this choice of lower bound is not important, but the knowledge of some lower bound is needed) and an upper bound of $H$ on agent values. The following more direct bound on $p_i$ in terms of the weakest type's value will be needed.

\begin{lemma}\label{lemma:payment_lb} Let $T_i$ be a predictor such that $\err_i(\btheta_{-i})\le 0$. Its weakest-type VCG price $p_i^{\WT}=\min_{\ttheta_i\in T_i(\btheta_{-i})}w(\ttheta_i,\btheta_{-i}) - \sum_{j\neq i}\theta_j[\alpha^*]$ satisfies $p_i^{\WT}\ge\ttheta_i[\alpha^*]$, where $\ttheta_i$ is the weakest type that minimizes $w(\ttheta_i,\btheta_{-i})$ over $\ttheta_i\in T_i(\btheta_{-i})$ and $\alpha^*$ is the efficient allocation achieving $w(\btheta)$.
\end{lemma}

\begin{proof}
    Let $\ttheta_i$ be the weakest type in $\Theta_i$ with respect to $\btheta_{-i}$. Since $\err_i\le 0$, weakest type prices are IR. The utility for agent $i$ is $$\begin{aligned}
    \theta_i[\alpha^*] - p_i^{\WT} 
    &= \sum_{j=1}^n\theta_j[\alpha^*] - \min_{\theta_i'\in T_i(\btheta_{-i})}\left(\max_{\alpha\in\Gamma}\sum_{j\neq i}\theta_j[\alpha] + \theta_i'[\alpha]\right) \\ 
    &= \sum_{j=1}^n\theta_j[\alpha^*] - \left(\max_{\alpha\in\Gamma}\sum_{j\neq i}\theta_j[\alpha] + \ttheta_i[\alpha]\right) \\ 
    &\le \sum_{j=1}^n\theta_j[\alpha^*] - \left(\sum_{j\neq i}\theta_j[\alpha^*] + \ttheta_i[\alpha^*]\right) \\
    &= \theta_i[\alpha^*] - \ttheta_i[\alpha^*],
\end{aligned}$$ so $p_i^{\WT}\ge \ttheta_i[\alpha^*]$, as desired. 
\end{proof}

We now describe the formal ingredients of our mechanism. For each $i$, the mechanism designer knows that $\theta_i$ lies in a $k$-dimensional subspace $U_i=\text{span}(u_{i,1},\ldots,u_{i,k})$ of $\R^{\Gamma}$ where each $u_{i,j}\in\R^{\Gamma}_{\ge 0}$ lies in the non-negative orthant and $\{u_{i,1},\ldots, u_{i,k}\}$ is an orthonormal basis for $U_i$ ($U_i$ can depend on $\btheta_{-i}$). For simplicity, assume $H = 2^a$ for some positive integer $a$. Let $\cL_{i,j} = \left\{\lambda u_{i,j} : \lambda\ge 0\right\}\cap [0, H]^{\Gamma}$ be the line segment that is the portion of the ray generated by $u_{i,j}$ that lies in $[0, H]^{\Gamma}$. Let $y_{i,j}$ be the endpoint of $\cL_{i,j}$ with $\lVert y_{i,j}\rVert_{\infty} = H$ (the other endpoint of $\cL_{i,j}$ is the origin). Let $z_{i,j}^1 = y_{i,j}/2$ be the midpoint of $\cL_{i,j}$, and for $\ell = 2,\ldots, \log_2 H$ let $z_{i,j}^{\ell}=z_{i,j}^{\ell-1}/2$ be the midpoint of $\overline{0z_{i,j}^{\ell-1}}$. So $\lVert z_{i,j}^{\log_2 H}\rVert_{\infty} = 1$. We terminate the halving of $\cL_{i,j}$ after $\log_2H$ steps due to the assumption that $\theta_i\in [1,H]^{\Gamma}$. For every $k$-tuple $(\ell_1,\ldots, \ell_k)\in\{1,\ldots, \log_2H\}^k$, let $$\ttheta_i(\ell_1,\ldots, \ell_k) = \sum_{j=1}^kz_{i,j}^{\ell_j}.$$ Furthermore, let $$W_{\ell} = \left\{(\ell_1,\ldots, \ell_k)\in\left\{1,\ldots,\log_2 H\right\}^k : \min_j\ell_j = \ell\right\}.$$ The sets $W_1,\ldots, W_{\log_2 H}$ form a partition of $\{1,\ldots,\log_2 H\}^k$ into levels, where $W_{\ell}$ is the set of points with $\ell_{\infty}$-distance $H/2^{\ell}$ from the origin. 

%Our mechanism is the following modification of weakest-type VCG, which we denote by $\cM_{k}$.

\begin{figure}[t]
\begin{tcolorbox}[colback=white, parbox=false]
\underline{Mechanism $\cM_{k}$}\\
Input: {\em correct} subspace predictions $U_1(\btheta_{-1}),\ldots, U_n(\btheta_{-n})$.
    \begin{itemize}
        \item Agents asked to reveal types $\theta_1,\ldots,\theta_n$.
        \item Let $\alpha^* = \argmax_{\alpha\in\Gamma}\sum_{i=1}^n\theta_i[\alpha]$ be the efficient allocation. For each agent $i$ let $$p_i = w\left(\ttheta_i(\ell_{i, 1},\ldots, \ell_{i,k}),\btheta_{-i}\right) - \sum_{j\neq i}\theta_j[\alpha^*]$$ where $\ell_i$ is chosen uniformly at random from the set $\{1,\ldots,\log_2 H\}$ and $(\ell_{i,1},\ldots,\ell_{i,k})$ is chosen uniformly at random from $W_{\ell_i}$.
        \item Let $\cI = \left\{i : \theta_i[\alpha^*] - p_i\ge 0\right\}.$ If agent $i\notin\cI$, $i$ is excluded and receives zero utility (zero value and zero payment).
    \end{itemize}
\end{tcolorbox}
\end{figure}
We now state and prove the welfare and revenue guarantees satisfied by $\cM_{k}$. In the proof, we use the notation $\theta_i\succeq\theta_i'$ to mean $\theta_i[\alpha]\ge\theta_i'[\alpha]$ for all $\alpha\in\Gamma$. Clearly, $\theta_i\succeq\theta_i'\implies w(\theta_i,\btheta_{-i})\ge w(\theta_i',\btheta_{-i})$.

\begin{theorem}\label{theorem:subspace}
    $\cM_k$ satisfies $\E[\text{welfare}]\ge\frac{\OPT}{\log_2 H}$ and $\E[\text{revenue}]\ge\frac{\OPT}{2k(\log_2 H)^{k}}.$
\end{theorem}

\begin{proof}
We have $\E[\text{welfare}]=\sum_{i=1}^n\theta_i[\alpha^*]\cdot\Pr(w(\theta_i,\btheta_{-i})\ge w(\ttheta(\ell_{i,1},\ldots,\ell_{i,k}),\btheta_{-i}))\ge\sum_{i=1}^n\theta_i[\alpha^*]\cdot\Pr(\ell_i = \log_2 H) = \frac{1}{\log_2 H}\cdot\OPT$ (since $\theta_i\succeq\ttheta_i(\log_2H,\ldots,\log_2H)$). The proof of the revenue guarantee relies on the following key claim: for each agent $i$, there exists $\ell_{i,1}^*,\ldots, \ell_{i,k}^*\in\{1,\ldots,\log_2H\}$ such that $\ttheta(\ell_{i,1}^*,\ldots,\ell_{i,k}^*)\succeq\frac{1}{2}\theta_i$. To show this, let $\theta_i^j$ denote the projection of $\theta_i$ onto $u_j$, so $\theta_i = \sum_{j=1}^k \theta_i^j$ since $\{u_{i,1},\ldots,u_{i,k}\}$ is an orthonormal basis. Let $\ell_{i, j}^* = \min\{\ell : \theta_i^j\succeq z_{i, j}^{\ell}\}$. Then, $z_{i, j}^{\ell_{i, j}^*}\succeq\frac{1}{2}\theta_i^j$, so $$\ttheta(\ell_{i, 1}^*,\ldots,\ell_{i, k}^*) = \sum_{j=1}^k z_{i, j}^{\ell_{i, j}^*}\succeq\sum_{j=1}^k\frac{1}{2}\theta_i^j = \frac{1}{2}\theta_i.$$ We now bound the expected payment. Let $\ell_{i}^* = \min_j\ell_{i, j}^*$. We have
\begin{align*}
    \E[p_i] &\ge\E\big[p_i\mid (\ell_{i, 1},\ldots,\ell_{i, k}) = (\ell_{i, 1}^*,\ldots,\ell_{i,k}^*)\big]\cdot\Pr\big((\ell_{i, 1},\ldots,\ell_{i, k}) = (\ell_{i, 1}^*,\ldots,\ell_{i,k}^*)\big) \\ 
    &= \frac{1}{|W_{\ell_i^*}|\log_2 H}\cdot \E\big[p_i\mid (\ell_{i, 1},\ldots,\ell_{i, k}) = (\ell_{i, 1}^*,\ldots,\ell_{i,k}^*)\big] \\ 
    &\ge \frac{1}{\log_2 H((\log_2 H)^{k} - (\log_2 H - 1)^{k})}\cdot\E\big[p_i\mid (\ell_{i, 1},\ldots,\ell_{i, k}) = (\ell_{i, 1}^*,\ldots,\ell_{i,k}^*)\big] \\ 
    &\ge \frac{1}{k(\log_2 H)^{k}}\cdot\E\big[p_i\mid (\ell_{i, 1},\ldots,\ell_{i, k}) = (\ell_{i, 1}^*,\ldots,\ell_{i,k}^*)\big]    
\end{align*} since the probability of obtaining the correct type $\ttheta(\ell_{i,1}^*,\ldots,\ell_{i,k}^*)$ can be written as the probability of drawing the correct ``level'' $\ell_{i}^*\in\{1,\ldots,\log_2 H\}$ times the probability of drawing the correct type within the correct level $W_{\ell_i^*}$. We bound the conditional expectation with Lemma~\ref{lemma:payment_lb}: $$\E\big[p_i\mid (\ell_{i, 1},\ldots,\ell_{i, k}) = (\ell_{i, 1}^*,\ldots,\ell_{i,k}^*)\big] \ge\ttheta_i(\ell_{i, 1}^*,\ldots,\ell_{i, k}^*)[\alpha^*] \ge \frac{1}{2}\cdot\theta_i[\alpha^*].$$ Finally, $$\E[\text{revenue}] = \sum_{i=1}^n\E[p_i] \ge \frac{1}{2k(\log_2 H)^{k}}\cdot\sum_{i=1}^n\theta_i[\alpha^*]=\frac{1}{2k(\log_2 H)^{k}}\cdot\OPT,$$ as desired. 
\end{proof}

$\cM_k$ can be viewed as a generalization of the (tight) $\log H$ revenue approximation in the single-item limited-supply setting that is achieved by a second-price auction with a uniformly random reserve price from $\{H/2, H/4,\ldots, 1\}$~\citep{Goldberg01:Competitive} (and it recovers that guarantee when $k=1$). Our results apply not only to auctions but to general multidimensional mechanism design problems such as the examples presented in Section~\ref{subsection:examples}.

\subsection{Revenue-optimal Groves mechanisms given a known prior}\label{sec:other_forms+myerson}

In this section we consider a textbook mechanism design setup wherein the mechanism designer has access to a joint prior distribution over a joint type space $\Theta\subseteq\bigtimes_{i=1}^n \R_{\ge 0}^{\Gamma}$ over the agents. We formulate the design of the Groves mechanism that maximizes expected revenue over the prior subject to no other constraints other than IC and IR (in particular efficiency is no longer a constraint as in Section~\ref{section:problem_formulation}). We show that the problem reduces to $n$ independent single-parameter optimization problems for each agent, though we do not pursue the issue of deriving a closed form/more explicit characterization. The optimization problem for each agent depends on that agent's weakest type. 
%Interestingly, in the single-item setting with independent, symmetric, regular (a distribution with cumulative density function $F:\R\to[0,1]$ and continuous probability density function $f$ is {\em regular} if the map $v\mapsto v- \frac{1-F(v)}{f(v)}$ is monotonically increasing in $v$) priors $F=F_1=\cdots=F_n$, our formulation reduces to Myerson's optimal auction~\citep{Myerson81:Optimal} which in this special case is a second-price auction with an appropriate reserve price.

For each agent $i$, the revealed type vector $\btheta_{-i}$ of all other agents induces a conditional distribution $D_i$ over agent $i$'s type. The mechanism designer can then optimize over that conditional distribution directly, and separately, for each agent. The payment-maximizing Groves mechanism can therefore be written as: \begin{equation}\label{eq:revoptgroves}h_i(\btheta_{-i}) = \argmax_{w\ge \min_{\ttheta_i\in\mathrm{supp}(D_i)} w(\ttheta_i, \btheta_{-i})}~\E_{\widehat{\theta}_i\sim D_i}\left[\left(\widehat{\theta}_i[\alpha^*(\widehat{\theta}_i,\btheta_{-i})] - w(\widehat{\theta}_i,\btheta_{-i}) + w\right)\mathbf{1}\left[w\le w(\widehat{\theta}_i,\btheta_{-i})\right]\right]\end{equation} where $\mathrm{supp}(D_i)$ is the support of $D_i$. It charges agent $i$ a payment of $h_i(\btheta_{-i}) - \sum_{j\neq i}\theta_j[\alpha^*]$, excluding agents who are charged more than their value from the final allocation (that is, agents for which $h_i(\btheta_{-i}) > w(\btheta)$). The weakest type in agent $i$'s type space is inherently baked into the optimization to compute $h_i(\btheta_{-i})$. Indeed, the welfare contributed by agent $i$ is lower bounded by $\min_{\ttheta_i\in\mathrm{supp}(D_i)}w(\ttheta_i,\btheta_{-i})$. If $D_i$ is supported on the entire ambient type space $\R_{\ge 0}^{\Gamma}$, the weakest type is the zero type, and welfare is lower bounded by $w(0,\btheta_{-i})$, so the expected payment extracted by $h_i(\btheta_{-i})$ is lower bounded by the expected vanilla VCG payment.

Finally, we remark that the single-parameter optimization problem to compute $h_i(\btheta_{-i})$ really only involves a two-dimensional joint distribution over $(\widehat{\theta}_i[\alpha^*(\widehat{\theta}_i,\btheta_{-i})], w(\widehat{\theta}_i,\btheta_{-i}))\in\R^2$, that is, the induced joint distribution over agent $i$'s value in the efficient allocation and the welfare she creates (rather than the full type distribution which is $|\Gamma|$ dimensional).

\paragraph{Sale of a single indivisible item to multiple bidders} We apply the above approach to the sale-of-a-single-item setting with independent, symmetric, and regular (a distribution with cumulative density function $F:\R\to[0,1]$ and continuous probability density function $f$ is {\em regular} if $\varphi(v) = v- \frac{1-F(v)}{f(v)}$ is monotonically increasing in $v$) prior distributions over bidders' values. We show that this recovers Myerson's revenue optimal auction~\citep{Myerson81:Optimal} which in this case is a second-price auction with reserve price $\varphi^{-1}(0)$. 

\begin{theorem}
    When there is a single indivisible item for sale and bidders' values are independently drawn from the same regular distribution, the revenue-optimal Groves' mechanism~\eqref{eq:revoptgroves} is globally revenue optimal.
\end{theorem}

\begin{proof}
We denote bidder $i$'s true and revealed value by $v_i$. In the single-item setting the efficient allocation gives the item to the highest bidder, so all other bidders receive and pay nothing. Let $i = 1$ be the index of the highest bidder and let $i = 2$ be the index of the second-highest bidder. Then, $w(\widehat{v}_1, \vec{v}_{-1}) = \max\{\widehat{v}_1, v_2\}$ and $\widehat{v}_1[\alpha^*(\widehat{v}_1, \vec{v}_{-1})] = \widehat{v}_1\cdot\mathbf{1}[\widehat{v}_1 > v_2]$ (we abbreviate $\alpha^*(\widehat{v}_1, \vec{v}_{-1})$ as just $\alpha^*$ in the following). Let $F$ denote the cumulative density function that is common to all bidders and let $f$ be its probability density function. We have
$$\begin{aligned}
h_1(\vec{v}_{-1}) &= \argmax_{w\ge v_2} \E_{\widehat{v}_1}\left[\left(\widehat{v}_1[\alpha^*] - w(\widehat{v}_1,\vec{v}_{-i}) + w\right)\cdot\mathbf{1}[w\le w(\widehat{v}_1, \vec{v}_{-1})]\right] \\
&= \argmax_{w\ge v_2} \E_{\widehat{v}_1}\left[\left(\widehat{v}_1[\alpha^*] - w(\widehat{v}_1,\vec{v}_{-i}) + w\right)\cdot\mathbf{1}[w\le w(\widehat{v}_1, \vec{v}_{-1})]\mid \widehat{v}_1\ge v_2\right]\cdot\Pr(\widehat{v}_1\ge v_2) \\ &\qquad\qquad\;\;+ \E_{\widehat{v}_1}\left[\left(\widehat{v}_1[\alpha^*] - w(\widehat{v}_1,\vec{v}_{-i}) + w\right)\cdot\mathbf{1}[w\le w(\widehat{v}_1, \vec{v}_{-1})]\mid \widehat{v}_1 < v_2\right]\cdot\Pr(\widehat{v}_1 < v_2) \\
&= \argmax_{w\ge v_2} \E_{\widehat{v}_1}\left[w\cdot\mathbf{1}[w\le\widehat{v}_1]\mid\widehat{v}_1\ge v_2\right]\cdot (1-F(v_2)) + \underbrace{(w-v_2)\cdot\mathbf{1}[w\le v_2]\cdot F(v_2)}_{=~0\text{ as } w\ge v_2} \\
&= \argmax_{w\ge v_2} w(1-F(w))
\end{aligned}$$ which is achieved at $h_1(\vec{v}_{-1}) = \max\{v_2, \varphi^{-1}(0)\}$ where $\varphi(w) = w - \frac{1 - F(w)}{f(w)}$ is the virtual value function. This is precisely a second-price auction with reserve price $\varphi^{-1}(0)$, which is equivalent to Myerson's optimal auction in this setting (symmetric, regular, and independent bidder priors). 
\end{proof}

It is clear that in general settings this approach does not yield the revenue optimal auction. Indeed, it is well known that the revenue optimal mechanism in multi-item settings is randomized~\citep{Conitzer02:Complexity,hart2013menu,hart2015maximal,Conitzer03:Applications} but the mechanism we present is deterministic. Furthermore, we limit ourselves to a subclass of Groves mechanisms which place a strong restriction on the set of allocations that can be realized---a bidder receives either her winning bundle in the efficient allocation or the empty bundle---while Myerson's revenue-optimal auction even in more general single-item settings can sell the item to a bidder other than the highest bidder.

\section{Beyond VCG: weakest-type affine-maximizer mechanisms}
\label{sec:amas}

An {\em affine-maximizer} (AM) mechanism~\citet{Roberts79:characterization} is a generalization of VCG that modifies the allocation and payments via agent-specific multipliers $\omega=(\omega_1,\ldots,\omega_n)\in\R_{\ge 0}$ and an allocation-based boost function $\tau:\Gamma\to\R_{\ge 0}$. We define the {\em weakest-type} AM in the following natural way (adopting the setup and notation of Section~\ref{sec:wcvcg}). 

The weakest-type AM parameterized by $\omega,\tau$ works as follows. Agents' types $\theta_1,\ldots,\theta_n$ are elicited, the allocation used is $$\alpha_{\omega,\tau} = \argmax_{\alpha\in\Gamma}\sum_{i=1}^n\omega_{i}\theta_i[\alpha] + \tau(\alpha),$$ and bidder $i$ is charged payment $$p_i = \frac{1}{\omega_i}\left(\min_{\ttheta_i: (\ttheta_i,\btheta_{-i})\in\bTheta}\left(\max_{\alpha\in\Gamma}\sum_{j\neq i}\omega_j\theta_j[\alpha] 
 + \omega_i\ttheta_i[\alpha]+\tau(\alpha)\right) - \left(\sum_{j\neq i}\omega_j\theta_j[\alpha_{\omega,\tau}]+\tau(\alpha_{\omega,\tau})\right)\right).$$ In a vanilla AM there is no minimization and no $\omega_i\ttheta_i[\alpha]$ term, so it is already clear that the weakest-type AM is a strict revenue improvement over the vanilla AM. We now generalize Theorem~\ref{theorem:rev_optimal} to the affine-maximizer setting.
 
 \begin{theorem}
    For any $\omega\in\R_{\ge 0}^n$ and $\tau:\Gamma\to\R_{\ge 0}$, the weakest-type AM parameterized by $\omega$ and $\tau$ is incentive compatible and individually rational. Furthermore, if $\bTheta$ is convex, it is revenue optimal among all incentive compatible and individually rational mechanisms that implement the allocation function $\btheta\mapsto\argmax_{\alpha\in\Gamma}\sum_{i=1}^n\omega_i\theta_i[\alpha] +\tau(\alpha)$.
 \end{theorem} 
 
 \begin{proof}
     The proof involves an identical application of revenue equivalence as in the proof of Theorem~\ref{theorem:rev_optimal}. The key property required is that weakest-type AMA payment leaves the ``affine" weakest type with zero utility, that is, the weakest type's IR constraint is binding, which is immediate from the payment formula. (Here, the weakest type is the type that minimizes the efficient {\em affine} welfare.) Therefore, any payment rule that generates strictly more revenue must violate individual rationality on some type profile of the form $(\ttheta_i,\btheta_{-i})$, where $\ttheta_i$ minimizes $w(\ttheta_i,\btheta_{-i})$ over all $\ttheta_i$ such that $(\ttheta_i,\btheta_{-i})\in\bTheta_i$. 
 \end{proof}
 
Let $\OPT(\omega,\lambda) = \sum_{i=1}^n\theta_i[\alpha_{\omega,\lambda}]$ be the welfare of the $(\omega,\lambda)$-efficient allocation. All of the guarantees satisfied by $\cM$ carry over to $\cM(\omega,\lambda)$, the only difference being the modified benchmark of $\OPT(\omega,\lambda)$. Of course, $\OPT(\omega,\lambda)\le\OPT$ is a weaker benchmark than the welfare of the efficient allocation. However, the class of affine maximizer mechanisms is known to achieve much higher revenue than the vanilla VCG mechanism. We leave it as a compelling open question to derive even stronger guarantees on mechanisms of the form $\cM(\omega,\lambda)$ when the underlying affine maximizer is known to achieve greater revenue than vanilla VCG. In any case, if one has a tuned high-revenue AM on hand~\citep{curry2023differentiable}, our techniques (weakest-type AM and randomization) can be appended as a post-processor to further improve revenue.

\section{Conclusions and future research}

We developed a versatile new methodology for multidimensional mechanism design that incorporates side information about agent types with the bicriteria goal of generating high social welfare and high revenue simultaneously. We designed mechanisms for a variety of side information formats. Our starting point was the {\em weakest-type VCG mechanism}, which generalized the mechanism of~\citet{krishna1998efficient}. A randomized tunable version of that mechanism achieved strong welfare and revenue guarantees that were parameterized by errors in the side information, and could be tuned to boost its performance. We additionally applied the weakest-type mechanism to three other forms of side information: predictions that could express uncertainty, agent types known to lie on low-dimensional subspaces of the ambient type space, and a textbook mechanism design setting where the side information is in the form of a known prior distribution over agent types. Finally, we showed how to generalize our main results to affine-maximizer mechanisms.

There are many new research directions that stem from our work. For example, how far off are our mechanisms from the welfare-versus-revenue Pareto frontier? The weakest-type VCG mechanism is one extreme point, but what does the rest of the frontier look like? One possible approach here would be to extend our theory beyond VCG to the larger class of affine maximizers (which are known to contain higher-revenue mechanisms)---we provided some initial ideas in Section~\ref{sec:amas} but that is only a first step. 

{\em Computation and practical auction design:} An important facet that we have only briefly discussed is computational complexity. The computations in our mechanism involving weakest types scale with the description complexity of $T_i(\btheta_{-i})$ (\emph{e.g.}, the number of constraints, the complexity of constraints, and so on). An important question here is to understand the computational complexity of our mechanisms as a function of the differing (potentially problem-specific) language structures used to describe the predictors $T_i(\btheta_{-i})$. (\citet{prasad2026weakest} have taken steps in this direction in the setting that the $T_i(\btheta_{-i})$ described by linear constraints.) In particular, the kinds of side information that are accurate, natural/interpretable, and easy to describe might depend on the specific mechanism design domain. Expressive bidding languages for combinatorial auctions have been extensively studied with massive impact in practice~\citep{Sandholm07:Expressive,Sandholm13:Very}. Can a similar methodology be developed for side information? Another important direction here is the exploration of the weakest type idea in the realm of mechanism design problems with additional practically-relevant constraints. Examples include 
core-selecting combinatorial auctions~\citep{Day07:Fair,Day08:Core,day2012quadratic,prasad2026weakest},
mechanism design with investment incentives~\citep{akbarpour2021investment}, obviously strategy-proof mechanisms~\citep{li2017obviously}, and other concrete market design applications like sourcing~\citep{Sandholm13:Very}, catch-share reallocation to prevent overfishing~\citep{bichler2019designing}, spectrum auctions~\citep{goetzendorff2015compact,leyton2017economics}, and electricity market design~\citep{cramton2017electricity}.

{\em Improved revenue when there is a known prior:} Another direction is to improve the revenue of the Bayesian weakest-type VCG mechanism of Krishna and Perry when there is a known prior over agents' types. Here, the benchmark would be efficient welfare in expectation over the prior. Krishna and Perry's mechanism uses weakest types with respect to the prior's support to guarantee efficient welfare in expectation, but its revenue could potentially be boosted significantly by compromising on welfare as in our random expansion mechanism. We took some first steps in Section~\ref{sec:other_forms+myerson}, but many open questions remain. For example, is there a closed form characterization of the revenue-maximizing Groves mechanism? Can those ideas be applied to revenue optimization of weakest-type affine maximizer mechanisms? Another direction here is to study the setting when the given prior might be inaccurate. Can our random expansion mechanism be used to derive guarantees that depend on the closeness of the given prior to the true prior? Such questions are thematically related to {\em robust mechanism design}~\citep{bergemann2005robust}. Another direction along this vein is to generalize our mechanisms to depend on a known prior over prediction errors. 

Finally, the weakest-type VCG mechanism is a strict improvement over the vanilla VCG mechanism, yet it appears to have been seldom studied or applied since its derivation by Krishna and Perry. The potential application of the weakest-type paradigm to other pressing questions in mechanism design and beyond is an exciting direction for future research.

\subsection*{Acknowledgements}
This work was supported by Office of Naval Research and the Vannevar Bush Faculty Fellowship [Grant ONR N00014-23-1-2876], the Army Research Office [Award W911NF2210266], the
Simons Foundation [Award MPS-SICS-00826333], the National Science Foundation [Grants CCF1733556, CCF-1910321, IIS-1901403, RI-1901403, RI-2312342, and SES-1919453], the National Institutes of Health [Grant A240108S001], and the Defense Advanced Research Projects Agency [Grant
HR00112020003]. The authors thank the anonymous reviewers for their valuable suggestions, which substantially improved this manuscript. Most of this research was performed while Siddharth Prasad was at Carnegie Mellon University.

\appendix

\section{A class of predictors that cannot generate any additional revenue for a welfare consistent mechanism}

The following proposition shows that predictors that make purely qualitative statements about the efficient allocation cannot be used by a welfare consistent mechanism to generate any revenue beyond that of VCG.
\begin{proposition}\label{prop:qualitative}
    Let $\widetilde{\Gamma}\subset\Gamma$ be any set of allocations. Let $$T_i(\btheta_{-i}) = \left\{\widehat{\theta}_i : \text{ the efficient allocation on } (\widehat{\theta}_i,\btheta_{-i}) \text{ is in }\widetilde{\Gamma}\right\}.$$ The revenue of any welfare consistent mechanism given predictors $\{T_i\}$ as input is no more than that of VCG.
\end{proposition}
\begin{proof}
    Let $\alpha\in\widetilde{\Gamma}$ and consider the type $\ttheta_i\in\R^{\Gamma}$ defined by $$\ttheta_i[\alpha] = w(0,\btheta_{-i}) - \sum_{j\neq i}\theta_j[\alpha]$$ and $\ttheta_i[\alpha']=0$ for all $\alpha'\in\Gamma\setminus\{\alpha\}$. The efficient allocation on type profile $(\ttheta_i,\btheta_{-i})$ is $\alpha$ and the efficient welfare is $w(0, \btheta_{-i})$. Indeed, $$\ttheta_i[\alpha] + \sum_{j\neq i}\theta_j[\alpha] = w(0, \btheta_{-i}) \ge \ttheta_i[\alpha'] + \sum_{j\neq i}\theta_j[\alpha']$$ for any other  allocation $\alpha'\in\Gamma$. So, $\ttheta_i\in T_i(\btheta_{-i})$ and generates the smallest possible efficient welfare $w(0,\btheta_{-i})$, which shows that $\min_{\ttheta_i\in T_i(\btheta_{-i})}w(\ttheta_i,\btheta_{-i}) = w(0, \btheta_{-i})$. Any perfectly consistent mechanism is IC, IR, and efficient on type domain $T_i(\btheta_{-i})$, and it is straightfoward to see that $T_i(\btheta_{-i})$ is connected, so by Theorem~\ref{theorem:rev_optimal} (see~\citet{prasad2025revenue} for an extension to connected sets, relaxing convexity) its revenue is at most that of vanilla VCG.\qed
\end{proof}
In particular, a predictor that correctly guesses the efficient allocation is too weak to guarantee that any given agent generates any additional welfare for the system. Some other examples of predictions that are of the form prescribed by Propostion~\ref{prop:qualitative} include $$T_i(\btheta_{-i}) = \{\widehat{\theta}_i : \text{agent } i \text{ wins bundle B in the efficient allocation}\}$$ or $$T_i(\btheta_{-i}) = \{\widehat{\theta}_i : \text{ items X and Y are sold separately in the efficient allocation}\}.$$ Predictions that make {\em quantitative} statements about the efficient welfare are useful in our framework. For example,  predictors $T_i(\btheta_{-i}) = \{\widehat{\theta}_i : w(\widehat{\theta}_i,\btheta_{-i}) = \OPT\}$ that guess the efficient welfare correctly generate welfare and revenue equal to $\OPT$ when plugged into Theorem~\ref{theorem:revenue_characterization}.

\small
\bibliography{references,dairefs}
\bibliographystyle{plainnat}

\end{document}